\documentclass[ecta,nameyear,draft]{econsocart}

\usepackage{amsmath,amssymb,blkarray,bm,dsfont}
\usepackage{tikz}
\usepackage{subfigure}
\usepackage[mathscr]{euscript}
\usepackage{caption}
\usepackage{colortbl} 
\usepackage{threeparttable}

\usepackage{booktabs}
\usepackage{multirow}
\usepackage{xcolor} 
\definecolor{deepjunglegreen}{rgb}{0.05, 0.3, 0.3}
\usepackage{float}
\usepackage{enumerate}
\usetikzlibrary{positioning}

\usetikzlibrary{patterns,intersections}

\RequirePackage[colorlinks,citecolor=blue,linkcolor=blue,urlcolor=blue,pagebackref]{hyperref}

\newcommand{\Apre}{A_{\texttt{pre}}}
\newcommand{\bpre}{b_{\texttt{pre}}}
\newcommand{\apost}{a_{\texttt{post}}}

\newcommand{\idM}{\mathcal{M}_{\texttt{ID}}} 
\newcommand{\idE}{\mathcal{E}_{\texttt{ID}}} 

\startlocaldefs
\theoremstyle{remark}
\newtheorem{theorem}{Theorem}[section]

\newtheorem{prop}{Proposition}[section]

\newtheorem{xmpl}{Example}
\newtheorem{asm}{Assumption}
\newtheorem{remark}{Remark}[section]
\newtheorem{proc}{Procedure}

\defcitealias{ADH10}{ADH} 
\defcitealias{AAHIW21}{AAHIW}
\defcitealias{HSS22}{HSS}
\defcitealias{CHT07}{CHT}

\endlocaldefs

\begin{document}

\begin{frontmatter}

\title{Synthetic Parallel Trends}

\begin{aug}
%
%
%
\author[id=au1,addressref={add1}]{\fnms{Yiqi}~\snm{Liu}\ead[label=e1]{yl3467@cornell.edu}}
\address[id=add1]{%
\orgdiv{Department of Economics},
\orgname{Cornell University}
}
\end{aug}

\vspace{-.5cm}
\small{November 8, 2025}\\
\href{https://drive.google.com/file/d/1MD1JSP1aNwMH1MtrSSLZH9HQFjY-bNlD/view?usp=sharing}{\small{Click here for the latest version.}}

\vspace{.4cm}
\support{
I am grateful to Levon Barseghyan, Francesca Molinari, and José Luis Montiel Olea for their mentorship. I also thank Jiafeng Chen, Xavier D'Haultfoeuille, Jacob Dorn, Hyewon Kim, Lihua Lei, Douglas Miller, Yaroslav Mukhin, Zhuan Pei, Alice Qi, Chen Qiu, Jonathan Roth, Xiaoxia Shi, Kyungchul Song, Jörg Stoye, Liyang Sun, Amilcar Velez, Jaume Vives-i-Bastida, Stefan Wager, and seminar participants at Cornell University and Stanford Data-Driven Seminar for valuable feedback.
}
\begin{abstract}
Popular empirical strategies for policy evaluation in the panel data literature---including difference-in-differences (DID), synthetic control (SC) methods, and their variants---rely on key identifying assumptions that can be expressed through a specific choice of weights $\omega$ relating pre-treatment trends to the counterfactual outcome. While each choice of $\omega$ may be defensible in empirical contexts that motivate a particular method, it relies on fundamentally untestable and often fragile assumptions. I develop an identification framework that allows for all weights satisfying a \textit{Synthetic Parallel Trends} assumption: the treated unit’s trend is parallel to a weighted combination of control units’ trends for a general class of weights. The framework nests these existing methods as special cases and is by construction robust to violations of their respective assumptions. I construct a valid confidence set for the identified set of the treatment effect, which admits a linear programming representation with estimated coefficients and nuisance parameters that are profiled out. In simulations where the assumptions underlying DID or SC-based methods are violated, the proposed confidence set remains robust and attains nominal coverage, while existing methods suffer severe undercoverage.
\end{abstract}


\end{frontmatter}

\newpage
\section{Introduction}
Learning treatment effects inevitably requires assumptions about unobserved counterfactual outcomes, e.g., what would have happened to the already-treated units in the absence of treatment. The program evaluation literature has developed various empirical strategies to address this fundamental challenge, often leveraging information from comparison groups, e.g., those not exposed to the treatment. In applications with panel data where units are observed across multiple time periods, information from both untreated units and pre-treatment time periods serves as a source of identifying power under assumptions that connect them to the counterfactual. 

In this paper, I show that the key identifying assumptions underlying widely used empirical strategies in the panel data literature---including difference-in-differences (DID) under the parallel trends assumption, synthetic control (SC) methods, and their variants such as two-way fixed effects (TWFE) and synthetic difference-in-differences (SDID)--- can all be expressed through a specific choice of population weights $\omega$ that relate pre-treatment trends to the counterfactual outcome. While each choice of $\omega$ may be defensible in different empirical contexts that motivate a particular method, it relies on fundamentally untestable and often fragile assumptions that can imply very different weights across methods.\footnote{This echoes the observation in \citet*{AAHIW21} that, at the estimation level, each of the DID, SC, and SDID estimators can be written as a weighted average difference in observed outcomes, with data-dependent weights that differ markedly across estimators (see their Figure 1).} I introduce an identification framework that allows for all weights satisfying a weaker identifying assumption, formally stated in Assumption \ref{asm:spt} and termed \textit{Synthetic Parallel Trends} (SPT): the trend of the treated unit is parallel to a weighted combination of control units' trends for a general class of weights. Under SPT, the counterfactual trend is identified by a weighted average of post-treatment control trends, with weights that reproduce the treated unit's trends in the pre-treatment periods. 
This framework nests  existing methods as special cases, as each of their identifying assumptions is associated with a particular selection from the set of valid weights under SPT. Hence, the proposed framework is by construction robust to a large class of violations of the identifying assumptions required by those existing methods.

I study properties of models consistent with SPT, starting with affine weights whose components sum to one but are otherwise unrestricted in sign or magnitude. In this setting, absent restrictions on the data generating process (DGP) other than SPT, a dichotomy occurs: the sharp identification region for the counterfactual is either trivial (i.e., it spans the entire real line) or a singleton. The latter case of point identification occurs if and only if the underlying DGP has a special low-rank property detailed in Proposition \ref{prop:idM-singleton}; even when this property is not explicitly imposed but holds true in the DGP, the identified set will automatically “detect” it and adapt to a singleton. Notably, the DGP implied by a TWFE model has this low-rank property. However, the “all-or-nothing” nature of the result hints at its fragility: as shown in Example \ref{eg:TWFE-stability}, this property breaks down under perturbations to the TWFE model, which also invalidate the TWFE estimand. 

This observation motivates combining the observed data with SPT and more credible assumptions to yield informative identification regions for the counterfactual, in the spirit of \citet{M89} and the subsequent partial identification literature; see \citet{T10, M20} for reviews.
I next impose a nonnegativity constraint on the weights, thereby making affine weights convex. Convexity is a standard assumption in the causal inference literature where estimands take a weighted average form, and is often motivated by its interpretability; see, for example, \citet{Abadie21} for the use of convex weights in SC methods, and \citet{dCdH23} and \citet*{RSBP23} for weighting heterogeneous treatment effects in the TWFE literature. This paper emphasizes a different motivation for imposing convexity: it is a source of identifying power. 
Convex combinations of post-treatment control trends effectively restrict the counterfactual trend to lie between their minimum and maximum, hence ruling out affine weights that return unbounded values and ensuring a nontrivial identified set if the control trends are bounded.\footnote{This idea is reminiscent of the bounded support assumption in \citet{M89}, 
though the convex-weight condition here does not imply a bound on the support of the outcome data, but rather a bound on its first moment.}
Under SPT with convex weights, there exists a time-invariant convex weight such that the treated unit’s trend lies in the convex hull of the control trends across all time periods. 
Such convex combination may not be unique in short panels, and SPT explicitly accounts for potential multiplicity of weights, yielding a partially identified set for the counterfactual that can be conveniently characterized by linear programs.

SPT under convex weights nests as special cases the parallel trends assumption and the convex weighting schemes underlying SC methods. Specifically, the DID estimand under the canonical two-timing-group parallel trends assumption extended to multiple periods is associated with the convex weight equal to 
each control unit’s population share among all controls (Section \ref{subsec:PT}); the expectation of the original SC estimator proposed by \citet*{ADH10} is associated with the convex weight that balances observed pre-treatment outcomes between the treated unit and control units as the number of pre-treatment periods grows (Section \ref{subsec:sc-ppm}). I also show that the probability limit of a variant of SC methods proposed by \citet*{AAHIW21} corresponds to the convex weight that balances the unobserved latent factors between the treated unit and control units (Section \ref{subsec:sc-latent-match}). This unifying perspective sheds light on how these popular empirical strategies point identify the counterfactual by selecting a specific convex weight consistent with SPT and  assuming that the counterfactual is generated by this chosen weight, even when multiple convex weights can reproduce the treated unit’s pre-treatment trends. The credibility of each selection is context-dependent and rests on ultimately untestable assumptions involving the unobserved counterfactual. 
By allowing flexible choices of weights, the proposed framework is robust to violations of each method's identifying assumptions, as long as the treated unit's trends remain a convex combination of control trends.

I provide a valid confidence set that covers the identified set for the treated unit’s treatment effect under SPT at any pre-specified confidence level. The proposed inference procedure builds on \citet{FS19} when panel data or repeated cross-sections are available within each aggregate unit. The key insight is that the linear program characterization of the identified set is equivalent to a system of moment equalities linear in the observed trends, which can be estimated by differences in sample means at the usual $\sqrt{n}$ rate from micro-level data. The resulting test statistic is a Hadamard directionally differentiable mapping of these estimated moments that profiles out the nuisance parameters---the weights $\omega$---whose dimensionality can be much larger than the scalar treatment effect of interest. This profiling-out step translates the original linear program constraints that need to be estimated into a new objective function quadratic in $\omega$ and subject to a known simplex constraint. The proposed method is valid, though only point-wise, under weaker assumptions than those required by existing methods and is applicable for a large class of inference problems involving linear programs with estimated coefficients and potentially many nuisance parameters beyond the specific problem studied in this paper.

I demonstrate the empirical value of SPT by revisiting the placebo study of \citet*{BDM04} and \citet{AAHIW21} using the Merged Outgoing Rotation Group Earnings Data from the Current Population Survey, which provides repeated cross-sections of weekly earnings for women from 1979 to 2018 across 50 states. In simulation designs where parallel trends or the identifying assumptions of SC-based methods are violated, the proposed confidence set under SPT remains robust and attains the nominal coverage, whereas existing methods suffer severe undercoverage if the specific identifying assumption on which they rely is violated.

\subsection{Related Literature}
A growing literature in partial identification explores relaxations of the parallel trends assumption. \citet{MP18} introduce the bounded variation assumption that allows the post-treatment counterfactual trend to differ from pre-treatment trends by at most a given amount. Extending this approach, \citet{RR23} develop a general identification and inference framework for a class of restrictions on differences in trends. \citet{BK23} reinterpret parallel trends as a restriction on selection bias and bound post-treatment bias by the minimum and maximum pre-treatment biases, yielding an identified set characterized by union bounds. With two control units, \citet*{YKHS24} bound the treated unit's trends by the minimum and the maximum of the two control trends across all time periods, obtaining a similar union bound. I contribute to this literature by providing a new and non-nested set of relaxations to parallel trends, leveraging across-unit variations that are stable over time in settings with multiple control units, a source of information unexplored by \citet{BK23} and \citet{YKHS24}. In addition, parallel pre-trends may still hold under SPT but not parallel post-trends, and therefore the approach proposed by \citet{MP18} and \citet{RR23} to bound violations of parallel post-trends using observed violations of parallel pre-trends may not apply. Finally, statistical inference methods for union bounds or the setting in \citet{RR23} do not apply under the SPT assumption. To address this, I develop a new inference procedure that delivers asymptotically valid confidence sets.

A distinct growing literature, initiated by \cite*{AG03} and further developed by \cite*{ADH10}, proposes to estimate counterfactual outcomes as weighted averages of control outcomes, where in later work the weights are typically obtained as the unique solutions to penalized regression problems \citep[e.g.,][]{DI16, AAHIW21, AL21, BFR22, IV23}. 
Statistical properties are usually derived under a linear factor model, with identification made implicit through model assumptions and the choice of penalty.
\citet*{SSMB22} justify such a model under distributional restrictions involving auxiliary covariates and within a sampling framework where aggregate outcomes are averages of more granular observations.
A complementary line of work in the proximal inference literature uses control outcomes as proxies and assumes the existence of bridge functions that account for all confounding \citep*[e.g.,][]{SLNHT23, QSMDT24, PT25}. 

In contrast, I impose identifying assumptions only on population means similar to those in the DID literature, and contribute to the broader effort to develop formal identification analysis for SC-type methods through a partial identification framework.
I highlight an additional motivation for using convex weights in the SC literature beyond interpretability: convexity alone provides identifying power, yielding informative bounds on the counterfactual. 
I then build on the sampling framework in \citet{SSMB22} to provide valid inference procedures for the bounds on the treatment effect of interest, an area that has received less attention than its DID counterpart. A related study is \citet{FR21}, who propose a sensitivity analysis to bound the bias of the SC estimator by the biases from using SC to predict post-treatment control outcomes. Their approach is non-nested with the SPT assumption and treats these SC estimates as known without accounting for statistical uncertainty.
Two contemporaneous papers, \citet*{RS25} and \citet*{SXZ25}, also exploit micro-level data within aggregate units for inference. \citet*{RS25} independently derive a similar DID-implied weight proportional to the control population shares as discussed in Section \ref{subsec:PT}. However, different from this paper, \citet*{RS25} focus on a regret analysis that assumes uniqueness of the SC weight and compares it with the DID-implied weight in terms of misspecification bias; their inference procedure constructs a confidence interval for the treatment effect by test-inverting a confidence set for the weights and then applying a Bonferroni correction. \citet{SXZ25} impose a population-mean version of the SC assumption combined with parallel trends to form a doubly robust estimand that is valid if either assumption holds. Instead, I introduce a unifying framework that nests the identifying assumptions underlying DID and SC-based methods, rather than assuming any one of them holds.

Finally, there is a large literature on inference for linear programs (LPs) with estimated coefficients. Existing methods have considered inference on the optimal value of an LP \citep*[e.g.,][]{FH15, CR24, G25, GM25}, the optimal solutions \citep*[e.g.,][]{HSS22}, and the existence of a feasible solution \citep*[e.g.,][]{CSS25}. However, these procedures either do not profile out nuisance parameters such as the $\omega$ weights, introduce perturbations to LPs that lead to conservativeness, or require assumptions on the rank of the constraint coefficient matrix that are difficult to verify in the context of this paper, as the coefficient matrix implied by the existing methods considered here may have arbitrary rank.\footnote{For example, the rank of the coefficient matrix implied by a TWFE model in Example \ref{eg:TWFE-stability} is at most $1$, failing the full rank condition of \citet{FH15, G25, GM25} and making it difficult to verify the stable rank condition of \citet{CSS25}.} The profiled test statistic of this paper does not require these rank assumptions and asymptotic normality of the estimated moments alone suffices to prove its validity. The cost of less assumptions is that this method is only point-wise valid and requires test inversion  of the scalar treatment effect parameter and bootstrap critical values, though each inversion is fast due to the profiling-out step and reformulation of the original LP with estimated constraints into a quadratic program with known simplex constraint.

\textbf{Outline}: Section \ref{sec:setup} introduces notation. Section \ref{sec:id} characterizes the identified set under SPT with affine and convex weights, and shows existing methods are special cases. Section \ref{sec:inf} details the inference method for constructing a valid confidence set for the identified set. Section \ref{sec:sim} presents a simulation study. Section \ref{sec:conclusion} concludes. Appendix \ref{appn:A} collects the main proofs, with auxiliary results presented in Appendix \ref{appn:B}.

\section{Setup and Notation}
\label{sec:setup}
I focus on settings where a binary treatment is implemented at the aggregate level, such as states or countries. To facilitate the introduction of notations, consider the empirical example from \citet[ADH henceforth]{ADH10}. At the end of the year 1988, California implemented Proposition 99, a tobacco control program that increased cigarette excise tax by 25 cents per pack starting in January 1989. The outcome of interest is the state-level per capita cigarette sales in packs, observed annually from 1970 to 1989 for California and 38 other control states.\footnote{\citetalias{ADH10} exclude states that have implemented similar tobacco tax raises and the District of Columbia from the pool of control units, leaving a total of 38 control states.} Denote each aggregate unit (e.g., state) as unit $k\in\{1,...,K\}$, where $K$ is the total number of units, and without loss of generality let $k=1$ be the treated unit (e.g., California). Denote time periods as $t\in\{1,...,T_0,T\}$, where $T_0$ is the last pre-treatment period and $T$ is the first post-treatment period, corresponding to 1988 and 1989, respectively.  I focus on the case with one treated unit and one post-treatment period for simplicity, but results in this paper can be extended to multiple treated units and post-treatment periods.\footnote{For example, if the parameter of interest is the average post-treatment effect of the treated, then one can group multiple treated units into one treated group and multiple post-treatment periods into one post-treatment block similar to the setup in Section \ref{subsec:sc-latent-match}; if instead the parameter of interest is the post-treatment, time-varying heterogeneous treatment effects across treated units, then one can impose Assumption \ref{asm:spt} for each post-treatment period and treated unit combination, and the same analysis applies.} As is common in the SC literature, I take the identity of the treated unit and treatment timing as given; see \cite{IV23} for an alternative framework where they are viewed as random. I also abstract away from additional covariates, which may be included by conditioning the outcome on observed covariates.

For the purpose of identification analysis, the unit of observation is an aggregate unit and the outcome of interest is a unit's population mean. I adopt the notation in \citet{SSMB22} and denote nonstochastic population means by $\mu$. Let $\big(\mu_t^k(1),\mu_t^k(0)\big)$ be the pair of potential aggregate outcomes for unit $k$ at time $t$ with and without the absorbing binary treatment, respectively; in the running example, they correspond to the year-$t$ \textit{expected} per capita cigarette sales in packs in each state $k$, with and without the tobacco control program. At the population level, the policymaker observes $\mu_t^k(0)$ for each control unit $k\geq2$ in all periods $t$; for the treated unit, $ \mu_T^1(1)$ and $\mu_t^1(0)$ for $t\leq T_0$ are observed. Implicit in this notation are stable-unit-treatment-value and no-anticipation assumptions, under which the observed aggregate outcome is given by
$\mu_t^k\equiv\mu_t^k(0)+\left(\mu_t^k(1)-\mu_t^k(0)\right)\mathds{1}\{k=1,t=T\}.$ The target parameter of interest is the treatment effect for the treated unit ($k=1$) in the post-treatment period $T$,
\begin{align}
    \tau\equiv\mu_T^1(1)-\mu_T^1(0).\label{eqn:tau}
\end{align}
Identification of the treatment effect $\tau$ thus boils down to identifying the counterfactual $ \mu_T^1(0)$, i.e., what would happen to the treated unit (e.g., California) in the post-treatment period $T$, had it not implemented the treatment (e.g., tobacco control program). To fix ideas, I provide three examples below that contextualize the meaning of $\big(\mu_t^k(1), \mu_t^k(0)\big)$.

\begin{xmpl}[Panel Data]
\label{eg:panel}
The DID literature commonly assumes access to balanced panel data where the same individual is observed across multiple periods \citep[e.g.,][]{CS21, dCdH20, G21, SZ20, W21}. In this case, the policymaker observes a random sample of $(Y_{i1},\dots,Y_{iT},D_{i1},...,D_{iT},G_i)$, where $Y_{it}\equiv D_{it}Y_{it}(1)+(1-D_{it})Y_{it}(0)$ is the realized time-$t$ outcome for person $i$ and $D_{it}$ denotes whether person $i$ has been treated by time $t$; $\big(Y_{it}(1),Y_{it}(0)\big)$ denotes the pair of \textit{stochastic} potential outcomes for person $i$ at time $t$, with and without the treatment, respectively; $G_i\in\{1,...,K\}$ denotes which aggregate unit (e.g., state) person $i$ comes from. Then $\mu_t^k(1)$ and $\mu_t^k(0)$ are, respectively, $\mathbb{E}[Y_{it}(1)|G_i=k]$ and $\mathbb{E}[Y_{it}(0)|G_i=k]$, which are expected individual potential outcomes of those from unit $k$.
\end{xmpl}

\begin{xmpl}[Repeated Cross-Sections]
\label{eg:rcs}
    In some scenarios, the same person may not be observed across all time periods. Instead, repeated cross-sectional data where different individuals are sampled in different time periods may be available. Adapting the notation commonly used in the DID literature with repeated cross-sections \citep*[e.g.,][]{SZ20, CS21, SXZ25, SX25}, I use $T_i\in\{1,...,T\}$ to denote the period in which person $i$ is observed. Then the individual-level potential outcome under treatment status $d\in\{0,1\}$ is $Y_i(d)\equiv \sum_{t=1}^TY_{it}(d)\mathds{1}\{T_i=t\}$. The policymaker observes a random sample of $(Y_i, D_i, T_i, G_i)$, where $D_i$ is the treatment status indicator and $Y_i=D_iY_i(1)+(1-D_i)Y_i(0)$ is the realized outcome. Then $\mu_t^k(1)$ and $\mu_t^k(0)$ are, respectively, $\mathbb{E}[Y_{i}(1)|G_i=k,T_i=t]$ and $\mathbb{E}[Y_{i}(0)|G_i=k, T_i=t]$, which are expected individual potential outcomes of those from unit $k$ observed in period $t$.
\end{xmpl}

\begin{xmpl}[Linear Factor Models]
\label{eg:lfm}
In the SC literature, a common assumption is that $\mu_t^k(0)$ has a unit-by-time interactive factor structure, and its sample analog is modeled as $\mu_t^k(0)$ plus a noise term; see Section \ref{sec:sc} for details on linear factor models. Directly assuming a functional form on these aggregate outcomes can be useful in settings without micro-level data, for example when the outcome of interest is gross domestic product.
\end{xmpl}

Examples \ref{eg:panel}-\ref{eg:lfm} each implicitly specify a sampling process on which the statistical uncertainty in estimating the treatment effect $\tau$ depends.
However, this is different from 
the question of whether $\tau$ can be identified, i.e., whether one can express the counterfactual $\mu_T^1(0)$ in terms of population quantities that the sample data imperfectly reveals.
In Section \ref{sec:id}, I introduce an identifying assumption whose generality and robustness lead to partial identification of $\tau$. I address statistical uncertainty in the estimation of $\tau$ in Section \ref{sec:inf}.

\textbf{Notation.} Unless noted otherwise, $\|\cdot\|$ denotes the Euclidean norm for vectors and the spectral norm for matrices (i.e., the largest singular value). For two sequences $a_n$ and $b_n$, $a_n\lesssim b_n$ means $a_n\leq c\cdot b_n$ for some constant $c>0$.

\section{Identification}
\label{sec:id}
Since $\mu_T^1(0)$ is unobserved, identifying it necessarily requires assumptions. In this section, I introduce a new yet intuitive identification assumption, which connects the treated unit and control units through the evolution of their untreated potential outcomes over time. I discuss what can be learned under this assumption in Section \ref{subsec:LP}, and show that both DID (Section \ref{subsec:PT}) and SC-based methods (Section \ref{sec:sc}) can be viewed as special cases that satisfy this assumption:
\begin{asm}[Synthetic Parallel Trends (SPT)]
\label{asm:spt}
    \textit{There exists a set of weights $\{\omega_k\}_{k=2}^{K}$ such that $\sum_{k=2}^K\omega_k=1$ and for all $t\in\{2,...,T\}$,
    \begin{align*}
        \sum_{k=2}^K\omega_k\left(\mu_t^k(0)-\mu_{t-1}^k(0)\right)=\mu_t^1(0)-\mu_{t-1}^1(0).
    \end{align*}
    }
\end{asm}

Assumption \ref{asm:spt} requires that, absent the treatment, the trends of the treated unit can be expressed as an affine combination of the control units' trends. Let $\Delta\mu_t^k(0)\equiv\mu_t^k(0)-\mu_{t-1}^k(0)$ denote the period-$t$ trend of unit $k$. In matrix notation, Assumption \ref{asm:spt} asks for the existence of an affine solution $\omega\in\mathbb{R}^{K-1}$ to the following linear equations:
\begin{align}
\label{eqn:sle}
    \underbrace{\begin{bmatrix}
        \Delta\mu_2^2(0)& \cdots &  \Delta\mu_2^K(0)\\
        \vdots & \ddots & \vdots\\
        \Delta\mu_{T_0}^2(0) ~~& \cdots &~~ \Delta\mu_{T_0}^K(0)
    \end{bmatrix}}_{\equiv \Apre\in\mathbb{R}^{(T_0-1)\times (K-1)}}\cdot\begin{bmatrix}
    \omega_2\\
    \vdots\\
    \omega_K
\end{bmatrix}=\underbrace{\begin{bmatrix}
        \Delta\mu_2^1(0)\\
        \vdots\\
        \Delta\mu_{T_0}^1(0)
        \end{bmatrix}}_{\equiv\bpre\in\mathbb{R}^{T_0-1}},
\end{align}
where  \eqref{eqn:sle} is a system of $(T_0-1)$ equations with $(K-1)$ unknowns that involves \textit{pre-treatment} trends only; each column of $\Apre$ stacks a control unit's pre-trends and $\bpre$ stacks the treated unit's pre-trends, all observed at the population level. Assumption \ref{asm:spt} further requires that the weights solving \eqref{eqn:sle} carry over to the post-treatment period, thereby identifying the treated unit's counterfactual trend $\Delta\mu_T^1(0)$ by weighted averages of control post-trends. One can explicitly encode the add-up constraint for the weight $\omega$  by adding a row of $1$s to both $\Apre$ and $\bpre$ in \eqref{eqn:sle}, under which multiplicity of solutions usually arises when $T_0<K-1$, i.e., when the number of pre-treatment periods is less than the number of control units, as in the California example from Section \ref{asm:spt} and later in the placebo study in Section \ref{sec:sim}. In this case,  the solution to \eqref{eqn:sle} is generally not unique and the counterfactual trend $\Delta\mu_T^1(0)$ is consequently set-identified, as shown in the following proposition:

\begin{prop}
\textit{
Under Assumption \ref{asm:spt}, the sharp identified set for $\Delta\mu_T^1(0)$ is
    \begin{align}
        \idM\equiv\bigg\{\apost'\omega:\omega\in\mathbb{R}^{K-1}, \boldsymbol{1}'\omega=1, \Apre\omega=\bpre\bigg\}\label{eqn:idset}
    \end{align}
    where $\boldsymbol{1}$ is the conforming vector with all elements equal to $1$, $\apost\equiv\left[\Delta\mu_T^2(0),\, \cdots, \,\Delta\mu_T^K(0)\right]'\in\mathbb{R}^{K-1}$ stacks control units' post-trends, and $\Apre$ and $\bpre$ are defined in \eqref{eqn:sle}.
}
\label{prop:id}
\end{prop}
In words, the identified set $\idM$ for the counterfactual trend $\Delta\mu_T^1(0)$ is given by a weighted average of control trends in the post-treatment period $T$, denoted by $\apost$, for affine weights $\omega$ that balance the pre-trends of the treated unit and of the control units via the equation $\Apre\omega=\bpre$. The counterfactual outcome $\mu_T^1(0)$ is then set-identified by adding $\mu_{T_0}^1(0)$ back to an element in the identified set for the counterfactual trend $\Delta\mu_T^1(0)$. 

\begin{remark}[Weight Selection by Penalized Regressions]
\label{rem:weight-pref}
    Using matrix algebra and properties of generalized inverses \citep[see, for example, Result 3.2.7 of][]{RD02}, the identification set in  \eqref{eqn:idset} can be equivalently written as 
    \begin{align}
        &\idM=\bigg\{\apost'\omega: \omega, u\in\mathbb{R}^{K-1},  \boldsymbol{1}'\omega=1, \omega=\Apre^-\bpre+(\boldsymbol{I}_{K-1}-\Apre^-\Apre)u\bigg\},\label{eqn:idset-inverse}
    \end{align}
    where $\Apre^-$ is a generalized inverse of $\Apre$ and $\boldsymbol{I}_{K-1}$ is the $(K-1)\times (K-1)$ identity matrix.
    One can interpret  (\ref{eqn:idset-inverse})  as decomposing the counterfactual $\apost'\omega$ into a component anchored by $\Apre^-$ that yields  $\apost' \Apre^-\bpre$, plus a deviation term $\apost'(\boldsymbol{I}_{K-1}-\Apre^-\Apre)u$ scaled by a free parameter $u\in\mathbb{R}^{K-1}$. 
    Suppose, for example, that the policymaker prefers weighting schemes with minimal $\ell_2$-norm, i.e., $\|\omega\|_2$ enters negatively into their utility function. Then they will choose from the set $\idM$ the counterfactual corresponding to $\omega^{\ell_2}\equiv \Apre^\dagger\bpre$, where $\Apre^\dagger$ is the unique Moore-Penrose pseudoinverse of $\Apre$, and $(\boldsymbol{I}_{K-1}-\Apre^\dagger\Apre)u$ absorbs any orthogonal deviations from the minimum $\ell_2$-norm solution that are still consistent with Assumption \ref{asm:spt}. This interpretation has a connection to the SC literature using penalized regressions, including LASSO and ridge regressions that penalize $\ell_1$ and $\ell_2$ norm of the weights, respectively, or a combination of both \citep[e.g.,][]{ASS18, AAHIW21, BFR21, BFR22, CWZ21, DI16, IV23}. Different penalty terms encode different preferences over the weighting schemes. This paper makes explicit that each choice of weights can point-identify a distinct value for the counterfactual, so the choice should be carefully justified within the specific empirical context. In cases where no strong economic reasoning motivates a particular selection, the paper takes the unifying perspective under which all weighting choices satisfying SPT might be consistent with the underlying DGP.
    
    It would also be interesting to reverse this line of reasoning and ask what preferences over weighting schemes would justify a particular value of counterfactual in $\idM$ that a policymaker might choose, but I leave this question for future research.
\end{remark}

\begin{remark}[Time Weights] 
\label{rem:time-weights-A1}
While empirical settings with short panels where $T_0<K-1$ are common, a similar analysis can be applied to settings where $T_0\geq K-1$, under which the system \eqref{eqn:sle} may be inconsistent. In this case, exploring the variation across time periods is more appealing, and the policymaker may consider an alternative assumption similar to Assumption \ref{asm:spt} but with weights across pre-treatment time periods instead of control units. The analysis of these time weights is similar to the analysis of the unit weights; see Appendix \ref{appn:B:time-weights} for more discussion on time weights.
\end{remark}

\subsection{Characterizing the Identified Set via Linear Programming}
\label{subsec:LP}
Observe that $\idM$ is a closed and convex interval in $\mathbb{R}$, and therefore can be written as 
$$
\idM=[\mu_\texttt{l}, \mu_\texttt{u}],
$$
where $\mu_\texttt{l}$ and $\mu_\texttt{u}$ are the optimal values of the following linear programs:
\begin{align}
    \mu_\texttt{l} = &\hspace{.5cm}\min \,\,\apost' \omega \hspace{4cm} \mu_\texttt{u} =\hspace{.5cm}\max\,\, \apost' \omega\hspace{.95cm}\notag\\
    &\text{s.t. } \Apre\omega = \bpre, ~\boldsymbol{1}'\omega=1  \hspace{3cm} \text{s.t. } \Apre\omega = \bpre, ~\boldsymbol{1}'\omega=1\label{eqn:LP}
\end{align}
Importantly,  Assumption \ref{asm:spt} is refutable by checking whether the linear programs in \eqref{eqn:LP} are feasible; see Remark \ref{rem:test-feasible} for a discussion on testing feasibility.

At first glance, Assumption \ref{asm:spt} may seem too weak to provide sufficient identification power that yields an informative interval, as it imposes no restriction on the sign nor magnitude of the weights. Remarkably, however, Assumption \ref{asm:spt} alone can still point identify the counterfactual if the underlying DGP produces a particular low-rank trend matrix, even when this property is not explicitly imposed:
\begin{prop}
    \label{prop:idM-singleton}
\textit{
    Consider the $T_0\times K$ matrix of all trends:
    \begin{align}
    \label{eqn:trend-mat}
        \begin{bmatrix}
            \bpre & \Apre\\
            \Delta\mu_T^1(0) & \apost'
        \end{bmatrix}=\begin{bmatrix}
            \Delta\mu_2^1(0) & \Delta\mu_2^2(0) & \cdots &  \Delta\mu_2^K(0)\\
            \vdots & \vdots & \ddots & \vdots\\
            \Delta\mu_{T_0}^1(0) & \Delta\mu_{T_0}^2(0) & \cdots &  \Delta\mu_{T_0}^K(0)\\
            \Delta\mu_T^1(0) & \Delta\mu_T^2(0) & \cdots &  \Delta\mu_T^K(0)
        \end{bmatrix}.
    \end{align}
    Under Assumption \ref{asm:spt}, the first column is an affine combination of the other $(K-1)$ columns. 
    Then, $\idM$ is a singleton if and only if the last row is an affine transformation of the other $T_0$ rows; otherwise, $\idM=\mathbb{R}$.\footnote{An affine combination is a linear combination with coefficients summing to $1$, and an affine transformation is a linear combination plus a constant shift.}
}
\end{prop}

Proposition \ref{prop:idM-singleton} shows that $\idM$ either provides no information or point identifies the counterfactual.  The geometric intuition is simple: the linear programs in \eqref{eqn:LP} have bounded optimal values if and only if $\apost$ is orthogonal to the set of $\omega$ satisfying  $\boldsymbol{1}'\omega=1$ and $\Apre\omega=\bpre$, which are hyperplanes defined by normal vectors $\boldsymbol{1}$ and rows of $\Apre$. This orthogonality happens if and  only if $\apost$ is a linear combination of these normal vectors, i.e., when $\apost$ is in their span, since the set of valid $\omega$ is formed by intersection of hyperplanes defined by these normal vectors. That $\apost$ is a linear combination of $\boldsymbol{1}$ and rows of $\Apre$ is exactly the necessary and sufficient condition for point identification stated in Proposition \ref{prop:idM-singleton}, which under Assumption \ref{asm:spt} implies that the counterfactual trend 
$\Delta\mu_T^1(0)$ is the same linear combination of $\boldsymbol{1}$ and $\bpre$. Even if this property is not imposed explicitly but holds true in the underlying DGP, it will be effectively “learned” by the identified set, which adaptively shrinks to a singleton. Figure \ref{fig:span} illustrates this geometrically.

\begin{figure}
    \centering
    \includegraphics[width=0.5\linewidth]{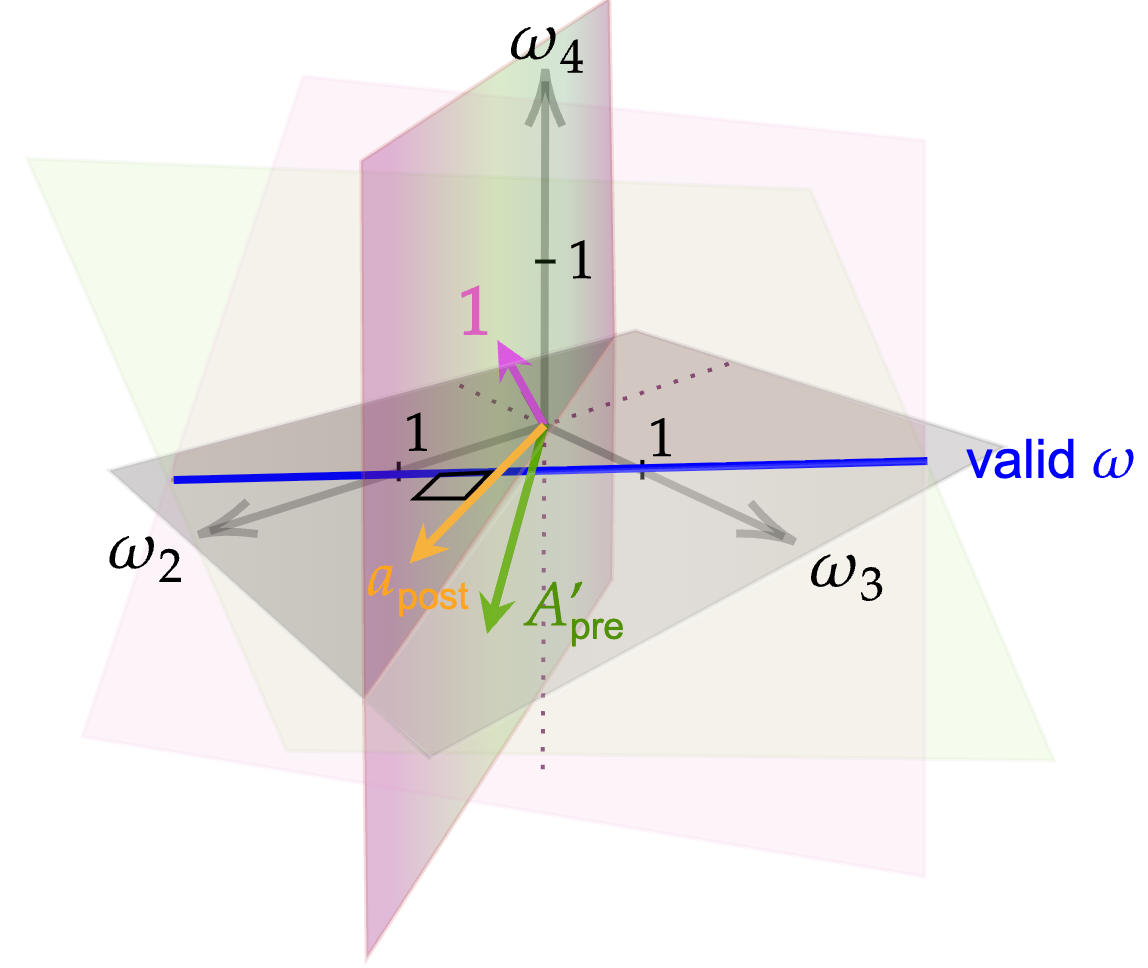}
    \caption{Necessary and sufficient condition for point identification under Assumption \ref{asm:spt}. With $T_0=2$ and $K-1=3$, $\Apre$ is a $1\times3$ row vector colored in green. The set of weights that are both affine (lying on the pink hyperplane defined by the normal vector $\boldsymbol{1}$) and satisfy $\Apre\omega=\bpre$ (lying on the green hyperplane defined by the normal vector $\Apre'$) is given by the intersection of these two hyperplanes, represented by the blue line. If $\apost$ (the yellow vector) is orthogonal to this intersection---that is, if $\apost$ is in the span of $\{\boldsymbol{1}, \Apre'\}$ (the hyperplane with a pink-to-green gradient)---then for any valid $\omega$, $\apost'\omega$ takes a finite value fixed by the distance from the blue line to the origin. Otherwise, the max and min of $\apost'\omega$ are unbounded along the blue line.}
    \label{fig:span}
\end{figure}

The class of DGPs that produce trends consistent with this low-rank structure includes those implied under a TWFE model, as shown in Example \ref{eg:TWFE-stability}.

\begin{xmpl}[Two-Way Fixed Effects]
    \label{eg:TWFE-stability}
    Suppose $\mu_t^k(0)$ is additively separable in time-$t$ and unit-$k$ fixed effects,
    \begin{align}
        \mu_t^k(0)=\lambda_t+\gamma_k,\quad\text{for all } 1\leq t\leq T, ~1\leq k\leq K,\label{eqn:TWFE}
    \end{align}
    where \eqref{eqn:TWFE} abstracts away from idiosyncratic shocks for the purpose of identification analysis. Then the period-$t$ trend $\Delta\mu_t^k(0)=(\lambda_t-\lambda_{t-1})$ is the same for all units $k$, implying that Assumption \ref{asm:spt} holds for any affine $\omega$. In this case, $\apost=[(\lambda_T-\lambda_{T_0}),\dots,(\lambda_T-\lambda_{T_0})]'$ has the same value for all its components and thus any affine $\omega$ will return the same $\apost'\omega=(\lambda_T-\lambda_{T_0})$, thereby point-identifying $\Delta\mu_T^1(0)=(\lambda_T-\lambda_{T_0})$. Note that in this example, since all units have the same trends in any time period, both $\apost$ and the rows of $\Apre$ are scalar multiples of $\boldsymbol{1}$, and it is straightforward to verify that $\apost$ is in the span of $\boldsymbol{1}$ and the row space of $\Apre$.
\end{xmpl}

Affine weights have been used in settings where homogeneity across units is imposed; see, for example, Remark 2.1 in \citet{CSX25}. In a similar spirit, the low-rank property in Proposition \ref{prop:idM-singleton} can be interpreted as homogeneity across time periods, under which post-trends are representable by pre-trends through an affine transformation. However, this property can easily break down. Suppose the underlying DGP deviates from the TWFE model in \eqref{eqn:TWFE} in a way that a unit-heterogeneous term $m_k$ enters only in the last period $T$:
\begin{align*}
    \mu_t^k(0)=\begin{cases}
        \lambda_t+\gamma_k & \text{if $t\leq T_0$,}\\
        \lambda_T+\gamma_k +m_k & \text{otherwise.}
    \end{cases}
\end{align*}
Then $\apost$ is no longer a scalar multiple of $\boldsymbol{1}$ anymore and falls outside of the span of $\boldsymbol{1}$ and $\Apre$, even if $\{m_k\}_{k=1}^K$ are infinitesimal in magnitude. In this case, $\idM=\mathbb{R}$ and what TWFE identifies can be arbitrarily wrong depending on the exact values of $\{m_k\}_{k=1}^K$.

Whether the trend matrix \eqref{eqn:trend-mat} exhibits the low-rank structure in Proposition \ref{prop:idM-singleton} has testable implications: one can test whether there exists $\varphi\in\mathbb{R}^{T_0}$ such that $\apost=[\boldsymbol{1}~~\Apre']\varphi$. This is similar to testing whether Assumption \ref{asm:spt} holds (or equivalently, whether the linear programs in \eqref{eqn:LP} are feasible); see Remark \ref{rem:test-feasible}. If the underlying DGP does not have this low-rank property, a natural next step is to consider more credible assumptions that restrict the counterfactual to a nontrivial, informative set. A common approach in the causal inference literature where estimands take a weighted average form is imposing nonnegativity on the weights. 
Combined with Assumption \ref{asm:spt}, this nonnegativity constraint ensures that the treated unit’s trend lies in the convex hull of the control trends across all time periods. 
While convexity is often motivated by interpretability and avoiding extrapolation, it also provides \textit{identifying power} by restricting the counterfactual trend to lie between the minimum and maximum of control post-trends, yielding an identified set that is a bounded subset of $\idM$ given by
\begin{align}
    \idM^+\equiv\bigg\{\apost'\omega:\omega\in\mathbb{R}_+^{K-1}, \boldsymbol{1}'\omega=1, \Apre\omega=\bpre\bigg\}=[\mu_\texttt{l}^+, \mu_\texttt{u}^+],\label{eqn:id-set+}
\end{align}
where $\mu_\texttt{l}^+$ and $\mu_\texttt{u}^+$ are the optimal values of the following linear programs with nonnegativity constraint $\omega\geq0$:
\begin{align}
        &\mu_\texttt{l}^+ = \hspace{.5cm}\min \,\,\apost' \omega \hspace{4cm} \mu_\texttt{u}^+ =\hspace{.5cm}\max\,\, \apost' \omega\hspace{.95cm}\notag\\
    &\hspace{1cm}\text{s.t.  \small $\Apre\omega = \bpre, ~\boldsymbol{1}'\omega=1, ~\omega\geq0$}  \hspace{2.2cm} \text{s.t.  \small $\Apre\omega = \bpre, ~\boldsymbol{1}'\omega=1, ~\omega\geq0$}\label{eqn:LP+}
\end{align}

The SC literature typically focuses on nonnegative weights for its interpretability and sparsity under certain conditions \citep[see][for a review]{Abadie21}. It should be emphasized that such nonnegativity constraint is an identifying assumption that shapes the identification region of the counterfactual. This is a refutable assumption if \eqref{eqn:LP+} is infeasible but \eqref{eqn:LP} is feasible. In particular, convexity may be rejected if the treated unit's trend is systematically larger or smaller than all control trends, causing $\Delta\mu_t^1(0)$ to lie outside of the convex hull of $\{\Delta\mu_t^2(0),\dots,\Delta\mu_t^K(0)\}$. 

In what follows, I draw connections between $\idM^+$ and widely-used empirical strategies.

\subsection{Connection to Difference-in-Differences under Parallel Trends}
\label{subsec:PT}
Assumption \ref{asm:spt} relaxes the canonical two-timing-group parallel trends assumption extended to multiple periods, which requires that, absent treatment, the outcome evolution of the eventually-treated group is the same as that of the never-treated group. Formally, for all $t\in\{2,...,T\}$, under the panel data setting in Example \ref{eg:panel},
\begin{align}
    \mathbb{E}[Y_{it}(0)-Y_{it-1}(0)|D_{iT}=1]=\mathbb{E}[Y_{it}(0)-Y_{it-1}(0)|D_{iT}=0],\label{eqn:pt}
\end{align}
and under the repeated cross-section setting in Example \ref{eg:rcs},
\begin{align}
    &\mathbb{E}[Y_i(0)|D_i=1,T_i=t]-\mathbb{E}[Y_i(0)|D_i=1,T_i=t-1]\notag\\
    =&\mathbb{E}[Y_i(0)|D_i=0,T_i=t]-\mathbb{E}[Y_i(0)|D_i=0,T_i=t-1].\label{eqn:pt-rcs}
\end{align}
Note that the TWFE model in \eqref{eqn:TWFE} differs from the two-timing-group parallel trends assumption in \eqref{eqn:pt}-\eqref{eqn:pt-rcs}: the former implicitly imposes a separate parallel trends assumption between the treated unit and each of the control units, whereas the latter defines comparison groups by treatment timing and pools all never-treated units into a single control group. The two coincide when there is only one control unit and one post-treatment period. Otherwise, the assumption imposed by TWFE is stronger, and as noted by \citet{CSX25}, implies \textit{over-identification} restrictions, which they leverage for efficiency gains when aggregate units are defined by treatment timing groups in a staggered-adoption setting.

Recall in Examples \ref{eg:panel}-\ref{eg:rcs}, $G_i\in\{1,...,K\}$ denotes the aggregate unit (e.g., state) to which individual $i$ belongs. Assume (i) in the panel data setting, $G_i$ does not vary across time (e.g., individual $i$ does not relocate over the sampling periods) and (ii) in the repeated cross-section setting, the share of observations from each control unit $k$ among all control observations is time-invariant, i.e., $\mathbb{P}(G_i=k|D_i=0,T_i=t)=\mathbb{P}(G_i=k|D_i=0)$ for all $k\geq 2$. If the policymaker believes that the parallel trends assumption \eqref{eqn:pt} holds under the panel data setting, then by the law of iterated expectation,
\begin{align}
    &\underbrace{\mathbb{E}[Y_{it}(0)-Y_{it-1}(0)|D_{iT}=1]}_{=\Delta\mu_t^1(0)}\notag\\
    =&\sum_{k=2}^K\underbrace{\mathbb{E}[Y_{it}(0)-Y_{it-1}(0)|D_{iT}=0, G_i=k]}_{=\Delta\mu_t^k(0)}\underbrace{\mathbb{P}(G_i=k|D_{iT}=0)}_{\equiv\omega_k^{\texttt{PT}}}.\label{eqn:PT-weights}
\end{align}
Therefore, the parallel trends assumption \eqref{eqn:pt} implicitly \emph{selects} a particular weighting scheme, namely $\omega^{\texttt{PT}}\equiv[\omega_2^{\texttt{PT}},...,\omega_K^{\texttt{PT}}]'$ defined in \eqref{eqn:PT-weights}, as the solution to the system of linear equations \eqref{eqn:sle}. These weights correspond to the population shares of each control unit among all control units, which are nonnegative and sum to $1$. 
An immediate implication of parallel trends is that the counterfactual trend identified by $\omega^{\texttt{PT}}$ should fall inside $\idM^+$: 
\begin{align}
    \underbrace{\mathbb{E}[Y_{iT}-Y_{iT_0}|D_{iT}=0]}_{=\apost'\omega^{\texttt{PT}}}\in\idM^+.\label{eqn:did-estimand}
\end{align} 
This is a testable implication, and testing whether \eqref{eqn:did-estimand} holds is conceptually equivalent to testing violations of parallel pre-trends commonly used in empirical research: observe that, by definition of $\idM^+$, \eqref{eqn:did-estimand} holds if and only if $\Apre\omega^{\texttt{PT}}=\bpre$, which is equivalent to the parallel trends assumption in \eqref{eqn:pt} holding in all pre-treatment periods.

An expression analogous to \eqref{eqn:PT-weights} can be derived given repeated cross-sectional (RCS) data, whose version of the parallel trends assumption \eqref{eqn:pt-rcs} implies
{\small{
\begin{align}
    &\underbrace{\mathbb{E}[Y_i(0)|D_i=1, T_i=t]-\mathbb{E}[Y_i(0)|D_i=1, T_i=t-1]}_{=\Delta\mu_t^1(0)}\label{eqn:pt-rcs-weight}\\
    =&\sum_{k=2}^K\underbrace{\big(\mathbb{E}[Y_i(0)|D_i=0, T_i=t, G_i=k]-\mathbb{E}[Y_i(0)|D_i=0, T_i=t-1, G_i=k]\big)}_{=\Delta\mu_t^k(0)}\underbrace{\mathbb{P}(G_i=k|D_i=0)}_{\omega_k^{\texttt{PT,RCS}}}.\notag
\end{align}}}And a similar analysis can be applied to the repeated cross-section setting. In what follows, I focus on the panel data setting and $\omega^\texttt{PT}$ in \eqref{eqn:PT-weights} for simplicity.

The relaxation of parallel trends allowed by Assumption \ref{asm:spt} can be interpreted as a set of restrictions that regulate how much post-treatment violations of parallel trends can deviate from their pre-treatment counterparts through the lens of the partial identification framework proposed by \citet{MP18} and \cite{RR23}. Using the notation of \citet{RR23}, denote the difference in trends by $\delta$, where 
{\small
\begin{align*}
    \delta\equiv\begin{bmatrix}
        \delta_{\texttt{pre}}\\
        \delta_{\texttt{post}}
    \end{bmatrix}, ~\delta_{\texttt{pre}}\equiv\begin{bmatrix}
        \Delta\mu_2^1(0)-\sum_{k=2}^K\omega_k^{\texttt{PT}}\Delta\mu_2^k(0)\\
        \vdots\\
        \Delta\mu_{T_0}^1(0)-\sum_{k=2}^K\omega_k^{\texttt{PT}}\Delta\mu_{T_0}^k(0)
    \end{bmatrix}, ~\delta_{\texttt{post}}\equiv\Delta\mu_T^1(0)-\sum_{k=2}^K\omega_k^{\texttt{PT}}\Delta\mu_T^k(0),
\end{align*}}i.e., $\delta_{\texttt{pre}}$ and $\delta_{\texttt{post}}$ are, respectively, the pre- and post-treatment differences in trends. Then the set of possible violations of parallel trends allowed by Assumption \ref{asm:spt} (SPT) with nonnegative weights is given by
{\small
\begin{align}
    \Delta^{\texttt{SPT}}\equiv\left\{\delta=\underbrace{\begin{bmatrix}
        \sum_{k=2}^K(\omega-\omega^{\texttt{PT}})_k\Delta\mu_2^k(0)\\
        \vdots\\
        \sum_{k=2}^K(\omega-\omega^{\texttt{PT}})_k\Delta\mu_{T_0}^k(0)\\
        \sum_{k=2}^K(\omega-\omega^{\texttt{PT}})_k\Delta\mu_T^k(0)
    \end{bmatrix}}_{\equiv\beta(\omega)}:\omega\in\mathbb{R}_+^{K-1}, \boldsymbol{1}'\omega=1,\Apre\omega=\bpre\right\},\label{eqn:polyhedron}
\end{align}}which collects all possible differences between the trend of the never-treated group, $\sum_{k=2}^K\omega^{\texttt{PT}}\Delta\mu_t^k(0)$, and a point from the convex hull of control trends that exactly matches the treated unit’s pre-treatment trends. Parallel trends implies $0\in\Delta^{\texttt{SPT}}$, i.e., $\omega^{\texttt{PT}}$ produces a valid convex combination of control trends equal to the treated unit's trend in pre-periods. 

$\Delta^{\texttt{SPT}}$ is a new set of relaxations non-nested with the proposals in \citet{RR23} that use violations of parallel pre-trends to bound violations of parallel post-trends. In particular, under SPT one may have parallel pre-trends if there is a convex $\omega\neq\omega^{\texttt{PT}}$ such that the first $T_0$ elements of $\beta(\omega)$ in \eqref{eqn:polyhedron} are $0$ but a non-parallel post-trend such that the last element of $\beta(\omega)$ is non-zero, in which case the bounding approach of \citet{RR23} would imply no violation of parallel post-trends given that all pre-trends are parallel. In addition, the inference method of \citet{RR23} does not apply to $\idM^+$. Although $\Delta^{\texttt{SPT}}$ is polyhedral in $\delta$---i.e., it can be written as $\{\delta:B\delta\leq \beta\}$ for matrix $B=[\boldsymbol{I}_{T_0}; -\boldsymbol{I}_{T_0}]'$, where $\boldsymbol{I}_{T_0}$ is the $(T_0\times T_0)$ identity matrix, and vector $\beta=[\beta(\omega)'; -\beta(\omega)']'$ for $\beta(\omega)$ defined in \eqref{eqn:polyhedron}---the dependence of $\beta(\omega)$ on the nuisance parameter $\omega$ subject to the constraint $\Apre \omega = \bpre$ places the problem outside the scope of \citet{RR23}. Their inference method builds on \citet{ARP23} and requires that the variance of the moment constraints does not depend on the nuisance parameter. This is not the case here, since $\omega$ is multiplied by $\Apre$, which needs to be estimated, and thus the variance of the implied moments depends on $\omega$; see \eqref{eqn:moment-eq-compact}. In Section~\ref{sec:inf}, I propose an alternative inference method.

\subsection{Connection to Synthetic Control Methods}
\label{sec:sc}
In the SC literature, a common assumption is that, absent the treatment, the \textit{sample} realization of the population outcome $\mu_t^k(0)$ that the policymaker would have observed, denoted by $\mu_t^{k,\texttt{s}}(0)$---with the superscript “$\texttt{s}$” indicating it is a sample quantity and thus stochastic---follows a linear factor model:
\begin{align}
    \mu_t^{k,\texttt{s}}(0)=\lambda_t'\gamma_k+\epsilon_{kt}, \label{eqn:lfm}
\end{align}
where $\lambda_t\in\mathbb{R}^F$ is a vector of latent time-varying factors, $\gamma_k\in\mathbb{R}^F$ is a vector of latent unit-specific loadings, and $\epsilon_{kt}\in\mathbb{R}$ is a mean-zero exogenous shock.\footnote{The additively separable unit and time fixed effects model in Example~\ref{eg:TWFE-stability} is a special case---with some abuse of notation---corresponding to a time factor $[\lambda_t,~1]'$ and a unit loading $[1,~\gamma_k]'$, where both $\gamma_k$ and $\lambda_t$ are scalars.} In this case, the population counterpart of $\mu_t^{k,\texttt{s}}(0)$ is given by the structural component of the factor model (\ref{eqn:lfm}),
$$\mu_t^k(0)=\mathbb{E}\big[\mu_t^{k,\texttt{s}}(0)\big]=\mathbb{E}[\lambda_t'\gamma_k],$$
where the expectation is taken over the joint distribution of $(\lambda_t,\gamma_k,\epsilon_{kt})$. Below I discuss the identifying assumptions behind the original SC method of \citetalias{ADH10} and one of its many variants. They differ in the assumption of whether the unit weights balance sample aggregate outcomes $\mu_t^{k,\texttt{s}}(0)$---i.e., the “perfect pre-treatment match” assumption stated in Assumption \ref{asm:ADH}(i)---or latent factors and loadings $\lambda_t'\gamma_k$ only. The former assumption offers transparency, as sample aggregate outcomes are directly observed in the data, but is difficult to satisfy in practice, as discussed in Remark \ref{rem:PPM-fail}. In contrast, the latter assumption sidesteps this challenge by imposing conditions on the unobservables, but at the cost of reduced transparency and verifiability. In what follows, I refer to both time factors and unit loadings as “factors” for simplicity and assume they are bounded in magnitude.

\subsubsection{Synthetic Controls with Perfect Pre-Treatment Match}
\label{subsec:sc-ppm}
Providing the first formal results in the SC literature, \citetalias{ADH10} impose the following assumptions on the factor model \eqref{eqn:lfm}, where $\widehat{\omega}^{\texttt{SC}}$ in Assumption \ref{asm:ADH}(i) denotes the SC weights, with the hat highlighting the fact that these weights depend on sample aggregate outcomes and are therefore stochastic:

\begin{asm}[\citetalias{ADH10}-SC Assumptions]
    \label{asm:ADH}
    \textit{
    \begin{enumerate}[(i)]
        \item There exists a set of convex weights $\widehat{\omega}^{\texttt{SC}}=\left[\widehat{\omega}_2^{\texttt{SC}},...,\widehat{\omega}_K^{\texttt{SC}}\right]'\geq 0$ such that $\boldsymbol{1}'\widehat{\omega}^{\texttt{SC}}=1$ and for all pre-treatment periods $t\leq T_0$,
        \begin{align}
            \mu_t^{1,\texttt{s}}(0)=\sum_{k=2}^K\widehat{\omega}_k^{\texttt{SC}}\mu_t^{k,\texttt{s}}(0).\label{eqn:PPM}
        \end{align}
        \item $\epsilon_{kt}$ is independent across units and time periods with $\mathbb{E}[\epsilon_{kt}]=0$ for all $k\in\{1,...,K\}$ and $t\in\{1,...,T\}$; for $k\geq 2$ and $t\leq T_0$, $\mathbb{E}[|\epsilon_{kt}|^p]<\infty$ for some even $p>2$; the smallest eigenvalue of $\frac{1}{T_0}\sum_{t=1}^{T_0}\lambda_t\lambda_t'$ is bounded below by $\underline{\xi}>0$; $|\lambda_{tf}|\leq\overline{\lambda}<\infty$ for all $f\in\{1,...,F\}$ and $t\in\{1,...,T\}$.
    \end{enumerate}
    }
\end{asm}

Under the linear factor model \eqref{eqn:lfm} and Assumption \ref{asm:ADH}, \citetalias{ADH10} show that the set of weights $\widehat{\omega}^{\texttt{SC}}$ in Assumption \ref{asm:ADH}(i) that balances pre-treatment sample aggregate outcomes between the treated unit and the control units (“perfect pre-treatment match” henceforth) carries over to the post-treatment period $T$, had unit $1$ not implemented the treatment:
\begin{align}
     \text{as }\,\, T_0\to\infty,\quad\left|\mathbb{E}\left[\sum_{k=2}^K\widehat{\omega}_k^{\texttt{SC}}\mu_T^{k,\texttt{s}}\right]-\mu_T^1(0)\right|\to 0,\label{eqn:scunbiased}
\end{align}
i.e., the \citetalias{ADH10}-SC estimator $\sum_{k=2}^K\widehat{\omega}_k^{\texttt{SC}}\mu_T^{k,\texttt{s}}$ is asymptotically unbiased for the counterfactual $\mu_T^1(0)=\mathbb{E}[\lambda_T'\gamma_1]$ as the number of pre-treatment periods $T_0$ grows. The next result shows that the identifying assumption of \citetalias{ADH10} is a special case of SPT with convex weights, though in an asymptotic sense as the \citetalias{ADH10}-SC method is only asymptotically valid. In the setting where $T_0$ grows, the dimension of the $(T_0-1)\times(K-1)$ pre-trends matrix $\Apre$ changes along the sequence $T_0\to\infty$. The following regularity assumption on the asymptotic behavior of $\Apre$ guarantees the identified set $\idM^+$ in \eqref{eqn:id-set+} is eventually nonempty and that the spectral norm of its Moore-Penrose inverse, $\|\Apre^\dagger\|$, does not diverge.

\begin{asm}[Asymptotic Behavior of $\Apre$]\label{asm:asymp-Apre}
   \textit{ Along the sequence $T_0\to\infty$, eventually the system $\Apre\omega=\bpre$  has a convex solution and the smallest singular value of $\Apre$ is bounded below by a strictly positive constant.}
\end{asm}

\begin{prop}
\label{prop:scid}
    \textit{
    Let Assumption \ref{asm:ADH} hold with $\mathbb{E}[|\epsilon_{1t}|^p]<\infty$ for some even $p>2$ and $t\leq T_0$.\footnote{While \cite{ADH10} only require “$\mathbb{E}[|\epsilon_{kt}|^p]<\infty$ for some even $p>2$” to hold for the control units ($k\geq 2$) as in Assumption \ref{asm:ADH}(ii), their main goal is to show \textit{post-treatment} asymptotic unbiasedness of the SC estimator, where the period-$T$ shock of the treated unit $\epsilon_{1T}$ enters the bias term with mean-zero and thus only the control shocks remain; see their $R_{2t}$ term on p. 504. On the other hand, Proposition \ref{prop:scid} shows the validity 
    of the weights in terms of $L^1$-norm, where pre-treatment shocks of the treated unit show up in the bias (see Eq. \eqref{eqn:pthmoment-trt}), and therefore the bounded $p$-th moment condition is also needed for $k=1$.} Then under deterministic factors $\{\lambda_t'\gamma_k\}_{t\in[T],k\in[K]}$,  as $T_0\to\infty$, 
    \begin{align*}
        \Apre\mathbb{E}\left[\widehat{\omega}^{\texttt{SC}}\right]-\bpre\,\xrightarrow[]{}\, \boldsymbol{0}\quad \text{and} \quad \apost' \mathbb{E}\left[\widehat{\omega}^{\texttt{SC}}\right]\,\xrightarrow[]{}\, \Delta\mu_T^1(0)=(\lambda_T-\lambda_{T_0})'\gamma_1,
    \end{align*}
    implying that under Assumption \ref{asm:asymp-Apre}, $\inf_{\delta\in\idM^+}\big|(\lambda_T-\lambda_{T_0})'\gamma_1-\delta\big|\to 0$. 
    }
\end{prop}

The proof of Proposition \ref{prop:scid} requires nonrandom factors $\lambda_t'\gamma_k$ to show that the \citetalias{ADH10}-SC weight $\widehat{\omega}^{\texttt{SC}}$ in \eqref{eqn:PPM} asymptotically satisfies Assumption \ref{asm:spt} by separating the nonstochastic part $\mathbb{E}[\lambda_t'\gamma_k]=\mu_t^k(0)$ from the expectation of the stochastic weighted sum $\mathbb{E}\big[\sum_{k=2}^K\widehat{\omega}_k^{\texttt{SC}}\lambda_t'\gamma_k\big]$. Deterministic factors are commonly assumed in the SC literature, \citep[e.g.,][]{AAHIW21, BFR21, F21, SBF25}, under which the expectation of the \citetalias{ADH10}-SC weights asymptotically recovers the counterfactual trend under the factor model, whose distance to the identified set under SPT with convex weights converges to $0$. Importantly, regardless of whether the SC assumptions are satisfied, $\idM$ and $\idM^+$ are well-defined objects under Assumption \ref{asm:spt} and do not rely on any particular functional form of the outcome.

\begin{remark}[Violation of Perfect Pre-Treatment Match]
\label{rem:PPM-fail}
    Proposition \ref{prop:scid} shows that the \citetalias{ADH10}-SC weight $\widehat{\omega}^{\texttt{SC}}$ asymptotically “solves” the system of equations $\Apre\omega=\bpre$ and recovers the latent factors of the treated unit. There, Assumption \ref{asm:ADH}(i)---perfect pre-treatment match---plays a key role in ensuring the validity of the \citetalias{ADH10}-SC estimator. However,
as also noted by \citetalias{ADH10} (p. 495), “it is often the case that no set of weights exists such that” Assumption \ref{asm:ADH}(i) is satisfied exactly in sample. In practice, $\widehat{\omega}^{\texttt{SC}}$ is estimated from data via solving the following constrained optimization problem, where $\Delta^{K-2}\equiv\{\omega\in\mathbb{R}_+^{K-1}:~\boldsymbol{1}'\omega=1\}$ denotes the simplex in $\mathbb{R}^{K-1}$:
\begin{align}
    \widehat{\omega}^{\texttt{SC}}\in\arg\min_{\omega\in\Delta^{K-2}}~~\sum_{t=1}^{T_0}\left(\mu_t^{1,\texttt{s}}(0)-\sum_{k=2}^K\omega_k\mu_t^{k,\texttt{s}}(0)\right)^2.\label{eqn:scobj}
\end{align}
For $\widehat{\omega}^{\texttt{SC}}$ to achieve perfect pre-treatment match, the objective function (\ref{eqn:scobj}) needs to be exactly $0$ at $\widehat{\omega}^{\texttt{SC}}$, a condition that often fails in practice \citep[see also][for a discussion on imperfect pre-treatment match]{FP21}. In addition, the bias of the \citetalias{ADH10}-SC estimator in \eqref{eqn:scunbiased} only vanishes in the limit as $T_0\to\infty$, yet perfect pre-treatment match becomes even more demanding as the number of pre-periods increases. To see this under Assumption \ref{asm:ADH} only,
let $\Lambda\equiv[\lambda_1,...,\lambda_{T_0}]$ and $\Gamma\equiv[\gamma_1,...,\gamma_K]$ stack the latent pre-treatment time and unit factors, respectively. Define the event that perfect pre-treatment match is satisfied in a particular $t\leq T_0$:
\begin{align*}
    E_t\equiv\left\{\exists\omega\in\Delta^{K-2}: \mu_t^{1,\texttt{s}}(0)=\sum_{k=2}^K\omega_k\mu_t^{k,\texttt{s}}(0)\right\}.
\end{align*}
Then the probability that, conditional on $(\Lambda, \Gamma)$, the perfect pre-treatment match assumption holds decreases exponentially in  $T_0$ under any non-degenerate linear factor model:
\begin{align*}
    \text{$\mathbb{P}\left(\exists \omega\in\Delta^{K-2}: \mu_t^{1,\texttt{s}}(0)=\sum_{k=2}^K\omega_k\mu_t^{k,\texttt{s}}(0)\,\,\,\forall\,t\leq T_0\biggr|\Lambda,\Gamma\right)\leq\prod_{t=1}^{T_0}\mathbb{P}\left(E_t~\biggr|~\bigcap_{\Tilde{t}=1}^{t-1} E_{\Tilde{t}}, \Lambda,\Gamma\right)$}
\end{align*}
where the inequality follows from applying the probability chain rule to $\mathbb{P}\big(\bigcap_{t=1}^{T_0}E_t\big)$. Note that the upper bound decreases exponentially in $T_0$ as long as  $\mathbb{P}\big(E_t~\big|~\bigcap_{\Tilde{t}=1}^{t-1} E_{\Tilde{t}}, \Lambda,\Gamma\big)=1$ for at most finitely many $t\leq T_0$, which essentially requires the period-$t$ shocks, $\{\epsilon_{1t},...,\epsilon_{Kt}\}$, are not perfectly predictable by past shocks and latent factors infinitely often, a mild condition for any non-degenerate factor model. 

Instead of the observed sample outcomes, the later SC literature has also considered a similar match assumption on the latent factors only, as described in the next section.
\end{remark}

\subsubsection{Synthetic Controls with Match on Latent Factors}
\label{subsec:sc-latent-match}
Rather than assuming perfect pre-treatment match on the observed sample outcomes, numerous papers in the SC literature consider the existence of weights that match directly on the latent factors \citep[e.g.,][]{AAHIW21, F21, FP21, IV23, SBF25}. In this section, I focus on the synthetic difference-in-differences (SDID) method proposed by \citet*[AAHIW henceforth]{AAHIW21} for its close relevance to the current paper in terms of combining insights from both DID and SC. I show that a result analogous to Proposition \ref{prop:scid} holds for SDID.

\citetalias{AAHIW21} also impose the factor model \eqref{eqn:lfm}, but assume \textit{nonrandom} factors $\lambda_t'\gamma_k$ as in Proposition \ref{prop:scid}. They focus on a diverging number of both units and time periods: for a total of $N$ units, let $\mathcal{I}_1$ denote the set of $N_1$ indices associated with units that are treated after period $T_0$, and remain exposed to the treatment until period $T_{\texttt{f}}>T_0$, where the subscript “$\texttt{f}$” indicates that $T_\texttt{f}$ is the final number of periods observed, during which $\{1,...,T_0\}$ indexes the pre-treatment periods, and $\mathcal{T}_1\equiv\{T_0+1,...,T_\texttt{f}\}$ indexes the post-treatment periods with size $|\mathcal{T}_1|=T_1$. To accommodate this multiple-treated-unit, multiple-post-period framework, 
\citetalias{AAHIW21} extend the linear factor model \eqref{eqn:lfm} to incorporate nonstochastic heterogeneous treatment effects across both units and time periods so that the observed sample aggregate outcome follows
\begin{align}
    \mu_t^{j,\texttt{s}}=&\lambda_t'\gamma_j+\mathds{1}\{j\in \mathcal{I}_1, t\in\mathcal{T}_1\}\tau_{jt}+\epsilon_{jt},~~\text{for }j\in[N],\label{eqn:sdid-lfm}
\end{align}
i.e., as in the \citetalias{ADH10} model \eqref{eqn:lfm}, absent the treatment, the policymaker observes in the sample the structural term $\lambda_t'\gamma_j$ plus the noise $\epsilon_{jt}$; but for units $j$ in the treated group $\mathcal{I}_1$ after treatment exposure $t\in\mathcal{T}_1$, there is an additional treatment effect term $\tau_{jt}$. 

Let $\mathcal{I}_0\equiv [N]\setminus\mathcal{I}_1$ collect the set of $N_0$ control unit indices in ascending order. Re-index the control units in $\mathcal{I}_0$ by $\{2,...,K\}$, assigning index $k\in\{2,...,K\}$ to the $(k-1)$st smallest element of $\mathcal{I}_0$, so that index $2$ corresponds to the smallest element, index $3$ to the second smallest, and so on, with $K=N_0+1$ corresponding to the largest. \citetalias{AAHIW21} propose an asymptotic framework where either $T_1$ or $N_1$ can be non-diverging, but not both. To translate their potential outcome notation to the nomenclature of this paper, I follow their condensed-form notation \citepalias[Section VII.1,  Online Appendix]{AAHIW21} and group all treated units $j\in\mathcal{I}_1$ into a single treated group re-indexed by $k=1$, and collect all post-treatment periods $t\in\mathcal{T}_1$ into a single post-treatment block re-indexed by $t=T$. Specifically, let $\mu_{N_1,t}^1(0)\equiv\frac{1}{N_1}\sum_{j\in \mathcal{I}_1}\lambda_t'\gamma_j$ group the $N_1$ treated units in period $t$; before treatment $t\leq T_0$, 
\begin{align}
    \mu_t^k(0)&=\lambda_t'\gamma_k, ~~\text{ for $k\in\{2,...,K\}$},\notag\\
    \mu_t^1(0)&=\begin{cases}
    \mu_{N_1,t}^1(0), & \text{fixed $N_1$}\\
        \lim_{N_1\to\infty}\mu_{N_1,t}^1(0), & \text{diverging $N_1$}
    \end{cases}\label{eqn:sdid-trt-pre}
\end{align}
and after treatment exposure,
\begin{align}
    \mu_T^k(0)&=\begin{cases}
        \frac{1}{T_1}\sum_{t\in\mathcal{T}_1}\lambda_t'\gamma_k, &\text{fixed $T_1$}\\
        \lim_{T_1\to\infty}\frac{1}{T_1}\sum_{t\in\mathcal{T}_1}\lambda_t'\gamma_k, &\text{diverging $T_1$}
    \end{cases} ~~\text{ for $k\in\{2,...,K\}$},\notag\\
    \mu_T^1(0)&=\begin{cases}
    \lim_{T_1\to\infty}\frac{1}{T_1}\sum_{t\in\mathcal{T}_1}\mu_{N_1,t}^1(0), & \text{fixed $N_1$ but diverging $T_1$}\\
    \lim_{N_1\to\infty}\frac{1}{T_1}\sum_{t\in\mathcal{T}_1}\mu_{N_1,t}^1(0), & \text{fixed $T_1$ but diverging $N_1$}\\
    \lim_{T_1,N_1\to\infty}\frac{1}{T_1}\sum_{t\in\mathcal{T}_1}\mu_{N_1,t}^1(0), & \text{both $N_1$ and $T_1$ diverging}
    \end{cases}\label{eqn:sdid-trt-post}
\end{align}
For simplicity and without loss of generality, I focus on the asymptotic regime with a fixed post-treatment period ($T_1=1$ and $T_\texttt{f}$ coincides with $T$) and $N_1 \to \infty$; extending the results to the other asymptotic settings described above is straightforward but entails more cumbersome notation. Under this framework, the parameter of interest $\tau$ in \eqref{eqn:tau} is given by
\begin{align}
    \text{\small $\tau\equiv\mu_T^1(1)-\mu_T^1(0)=\lim_{N_1\to\infty}\frac{1}{N_1}\sum_{j\in\mathcal{I}_1}\tau_{jT}$, ~~~where $\mu_T^1(1)=\lim_{N_1\to\infty}\frac{1}{N_1}\sum_{j\in\mathcal{I}_1}(\lambda_T'\gamma_j+\tau_{jT})$}.\label{eqn:sdid-tau}
\end{align}
The limits in \eqref{eqn:sdid-trt-pre}-\eqref{eqn:sdid-tau} are assumed to exist and be finite. In addition to unit-specific weights, \citetalias{AAHIW21} also consider weighting across pre-treatment time periods and allow for a constant shift that can be differenced out. Let $\omega\equiv[\omega_2,...,\omega_K]'$ denote a set of convex unit weights and $\omega_0\in\mathbb{R}$ a constant. The following is a restatement of a subset of assumptions in Assumption 4 of \citetalias{AAHIW21} relevant for identification using unit weights:
\begin{asm}[SDID Identification]
\label{asm:SDID}
\textit{
Let $(\Tilde{\omega}_0,\Tilde{\omega})$ be the nonstochastic oracle weights defined in Section III-B of \citet[p.4104]{AAHIW21} that solves\footnote{In \eqref{eqn:sdid-obj}, $\zeta$ is a penalty parameter that depends on the number of pre-periods and covariance of the error term; see Eq. (19) of \citetalias{AAHIW21} for its exact definition.}
\begin{align}
    \min_{\omega_0\in\mathbb{R},\omega\in\Delta^{K-2}}\sum_{t=1}^{T_0}\left(\omega_0+\sum_{k=2}^K\omega_k\mu_t^k(0)-\mu_{N_1,t}^1(0)\right)^2+\zeta\|\omega\|_2^2.\label{eqn:sdid-obj}
\end{align}
As $T_0,K,N_1\to\infty$ and for $\nu\equiv[\nu_1,...,\nu_{T_0}]'$ a convex weight, the following holds:
    \begin{enumerate}[(i)]
        \item for $t\leq T_0$, $\big|\Tilde{\omega}_0+\sum_{k=2}^K\Tilde{\omega}_k\mu_t^k(0)-\mu_{N_1,t}^1(0)\big|\to0$;
        \item $\left(\mu_{N_1,T}^1(0)-\sum_{k=2}^K\Tilde{\omega}_k\mu_T^k(0)\right)-\sum_{t=1}^{T_0}{\nu}_t\left(\mu_{N_1,t}^1(0)-\sum_{k=2}^K\Tilde{\omega}_k\mu_t^k(0)\right)\to0$.
    \end{enumerate}
    }
\end{asm}

Assumption \ref{asm:SDID}(i) is a restatement of the equation immediately following Eq. (24) of \citetalias{AAHIW21}. It requires $(\Tilde{\omega}_0, \Tilde{\omega})$ to exactly minimize the first part of the objective function in \eqref{eqn:sdid-obj}, though in an asymptotic sense, to achieve perfect pre-treatment match on the latent factors up to a constant $\Tilde{\omega}_0$ difference as $T_0,K,N_1,\to\infty$; Assumption \ref{asm:SDID}(ii) is a restatement of Eq. (26) of \citetalias{AAHIW21} that ensures the oracle unit weights are also valid in the post-period $T$ to recover the treated unit's counterfactual factor $\mu_{N_1,T}^1(0)$ via a double-differencing form that cancels out the constant difference $\tilde{\omega}_0$. To see this, note that Assumption \ref{asm:SDID}(i) implies that the subtrahend in Assumption \ref{asm:SDID}(ii) goes to $\tilde\omega_0$, implying\footnote{For more algebraic details, see the proof of Proposition \ref{prop:scid} in Appendix \ref{appn:A}.}
\begin{align*}
   \mu_T^1(0)-\left(\Tilde{\omega}_0+\sum_{k=2}^K\Tilde{\omega}_k\mu_T^k(0)\right)\to0.
\end{align*}

The first part of the objective function in \eqref{eqn:sdid-obj} can be interpreted as an asymptotic version of SPT under convex weights: since the constant $\omega_0$ can be canceled out after first-order differencing, implying that the treated unit's trend is a convex combination of control units' trends in the pre-periods. However, such convex combination may not be unique, and the second part of the objective function in \eqref{eqn:sdid-obj} \textit{selects} the convex weight with the smallest $\ell^2$-norm. This selection may not be the correct one that recovers the counterfactual in the post-period $T$, but is nevertheless assumed to be correct under Assumption \ref{asm:SDID}(ii). Hence, Assumption \ref{asm:SDID} is a special case of SPT, which takes into account all convex weights that reproduce the treated unit's pre-trends, and a result analogous to Proposition \ref{prop:scid} holds:
\begin{prop}
\label{prop:sdid-id}
\textit{
    Let Assumption \ref{asm:SDID} hold. Then as $T_0, K, N_1\to\infty$,
    \begin{align*}
        \Apre\Tilde{\omega}-\bpre\to \boldsymbol{0} \quad\text{and}\quad \apost'\Tilde{\omega}\to\Delta\mu_T^1(0)=\lim_{N_1\to\infty}\frac{1}{N_1}\sum_{j\in \mathcal{I}_1}(\lambda_T-\lambda_{T_0})'\gamma_j.
    \end{align*}
     Under a strengthened version of Assumption \ref{asm:asymp-Apre} that accounts for $K, N_1\to\infty$, 
     \begin{align}
         \inf_{\delta\in\idM^+}\left|\Delta\mu_T^1(0)-\delta\right|\to0.\label{eqn:set-dist-sdid}
     \end{align}
}
\end{prop}
\begin{remark}[Probability Limit of the SDID Estimator] Let $\widehat{\tau}^{\texttt{SDID}}$ denote the SDID estimator for $\tau$ defined in Algorithm 1 of \citetalias{AAHIW21} (p. 4093). Then under additional assumptions stated in their Theorem 1, \citetalias{AAHIW21} show that $\widehat{\tau}^{\texttt{SDID}}$ is asymptotically normal and centered at the true $\tau$. An immediate implication is consistency: $\widehat{\tau}^{\texttt{SDID}}-\tau=o_p(1)$ for the expression of $\tau$ given in \eqref{eqn:sdid-tau} under the factor model. Then Proposition \ref{prop:scid} implies that the distance between the probability limit of 
$\widehat{\tau}^{\texttt{SDID}}$ and the identified set for the \textit{treatment effect} under SPT with convex weights, formally defined in \eqref{eqn:idE+}, converges to $0$. 
\end{remark}

In conclusion of Section \ref{sec:id}, the common ground of both DID under parallel trends and SC-based methods is the intuition that, absent treatment, the treated unit can be expressed as a convex combination of a group of control units. Assumption \ref{asm:spt} is an identifying assumption that synthesizes this idea of weighting at the population expectation level and nests widely-used empirical strategies under the conditions stated in this section.

\section{Inference on the Identified Set for the Treatment Effect}
\label{sec:inf}
So far the focus has been on identifying the counterfactual trend $\Delta\mu_T^1(0)\equiv\mu_T^1(0)-\mu_{T_0}^1(0)$, which then identifies the counterfactual outcome $\mu_T^1(0)$ and consequently the treatment effect $\tau\equiv\mu_T^1(1)-\mu_T^1(0)$. However, for inference, the sampling uncertainty in both $\mu_T^1(1)$ and $\mu_T^1(0)$ should be taken into account and thus inference for $\tau$ directly is more useful than inference for the counterfactual alone. Under Assumption \ref{asm:spt} with convex weights, the identified set for $\tau$ is given by
\begin{align}
    \idE^+\equiv\left\{\mu_T^1(1)-\mu:\mu-\mu_{T_0}^1(0)\in\idM^+\right\},\label{eqn:idE+}
\end{align}
i.e., $\idE^+$ is the set of values that takes the difference between the observed period-$T$ mean outcome $\mu_T^1(1)$ and a candidate value $\mu$ for the counterfactual $\mu_T^1(0)$, for the set of $\mu$ that is consistent with a counterfactual trend value $\mu-\mu_{T_0}^1(0)$ in the identified set $\idM^+$ under convex weights defined in \eqref{eqn:id-set+}.

In this section, I propose a valid confidence set that covers every element of $\idE^+$ with a pre-specified $(1-\alpha)$  probability by test inversion, given estimators of $(\Apre, \bpre,\apost)$. The key insight that motivates the proposed inference procedure is that $\idE^+$ can be characterized by moment equalities: $\Tilde{\tau}\in\idE^+$ if and only if there exists $\omega\in\Delta^{K-2}$ such that
\begin{align}
\begin{cases}
    \Apre\omega-\bpre=0\\
    \apost'\omega-\left(\mu_T^1(1)-\Tilde{\tau}-\mu_{T_o}^1(0)\right)=0
\end{cases}\label{eqn:moment-eq}
\end{align}
where the second equality in \eqref{eqn:moment-eq} follows from the identity $\mu_T^1(0)=\mu_T^1(1)-\tau$. Let
\begin{align}
    A\equiv\begin{bmatrix}
        \Apre\\
        \apost'
    \end{bmatrix}, ~~b\equiv\begin{bmatrix}
        \bpre\\
        \mu_T^1(1)-\mu_{T_0}^1(0)
    \end{bmatrix}\label{eqn:Ab-defn}
\end{align}
stack observed trends that need to be estimated. Let $vec(A)\equiv[A_{\cdot,1}',\dots,A_{\cdot,K-1}']'$ stack columns of $A$, $\boldsymbol{I}_{T_0}$ denote the $(T_0\times T_0)$ identity matrix, $\otimes$ be the Kronecker product, and $\boldsymbol{e}_{T_0}$ denote the $T_0$-th standard basis vector in $\mathbb{R}^{T_0}$. Rewrite the moment equalities \eqref{eqn:moment-eq} by
\begin{align}
m_{\tilde{\tau}}(\omega;A, b)\equiv\underbrace{[-\boldsymbol{I}_{T_0}, ~~\omega'\otimes\boldsymbol{I}_{T_0}]}_{\equiv J(\omega)}\begin{bmatrix}
    b\\
    vec(A)
\end{bmatrix}+\Tilde{\tau}\boldsymbol{e}_{T_0}=0,\label{eqn:moment-eq-compact}
\end{align}
where for a fixed $\tilde{\tau}\in\idE^+$, the moment function $m_{\tilde{\tau}}(\omega;A, b)$ is linear in both the nuisance parameter $\omega$ and $(A,b)$ to be estimated. Then the identified set for the treatment effect $\idE^+$ in \eqref{eqn:idE+} can be equivalently written as
\begin{align}
    \idE^+=&\left\{\tilde{\tau}\in\mathbb{R}:\,\exists\,\omega\in\Delta^{K-2} \text{ such that } m_{\tilde{\tau}}(\omega;A,b)=0\right\}\notag\\
    =&\left\{\tilde{\tau}\in\mathbb{R}: \min_{\omega\in\Delta^{K-2}}m_{\tilde{\tau}}(\omega;A,b)'m_{\tilde{\tau}}(\omega;A,b)=0\right\}\label{eqn:idE+criterion}
\end{align}
where \eqref{eqn:idE+criterion} translates the constraints of the original linear programs \eqref{eqn:LP+} that need to be estimated into an objective function quadratic in the choice variable $\omega$, which is then profiled out via $\min_{\omega\in\Delta^{K-2}}(\,\cdot\,)$ over a known simplex constraint $\Delta^{K-2}$. This transformation sidesteps the need to impose rank conditions on the constraint coefficient matrix $\Apre$ in the original linear programs \eqref{eqn:LP+} \citep[see, e.g.,][for examples of rank conditions on coefficient matrices]{FH15, G25, GM25, CSS25}. Allowing $\Apre$ to have arbitrary rank is important in the context of this paper, as the existing methods considered here may imply an $\Apre$ that is highly rank-deficient; in particular, the $\Apre$ matrix implied under the TWFE model in Example \ref{eg:TWFE-stability} has rank $\leq1$.

I assume the availability of either panel data as in Example \ref{eg:panel} or repeated cross-sections (RCS) as in Example \ref{eg:rcs}.

\begin{asm}[Sampling] 
\label{asm:sample}
\textit{ Assume a sample of size $n$ is available, where either
    \begin{enumerate}[(i)]
        \item (Panel) $\{Y_{i1},...,Y_{iT},G_i\}_{i=1}^n$ is independently and identically distributed (iid) drawn from $\mathbb{P}$, and $p_k\equiv \mathbb{P}(G_i=k) \geq \epsilon>0$ for constant $\epsilon$ and all $k\in\{1,...,K\}$ ; or
        \item (RCS) $\{Y_i,G_i, T_i\}_{i=1}^n$ for $(Y_i,G_i)$ iid sampled conditional on $T_i$ such that
        $\mathbb{P}(Y_i\leq y, G_i=k,T_i=t)=\pi_t\mathbb{P}_t(Y_i\leq y, G_i=k)$, where
        \begin{itemize}
            \item $\mathbb{P}_t$ is the conditional distribution of $(Y_i,G_i)$ given $T_i=t$ if $T_i$ is iid multinomial, in which case both $\pi_t=\mathbb{P}(T_i=t)$ and $\pi_{kt}\equiv\mathbb{P}(G_i=k,T_i=t)\geq\epsilon>0$; or
            \item $\mathbb{P}_t$ is the  distribution of $(Y_i,G_i)$ within stratum $T_i=t$ if $T_i$ is fixed, in which case both $\pi_t=\lim_{n\to\infty}\frac{\#(T_i=t)}{n}$ and $\pi_{kt}\equiv\pi_t\mathbb{P}_t(G_i=k)\geq \epsilon > 0$.
        \end{itemize}
    \end{enumerate}
}
\end{asm}

Assumption \ref{asm:sample} allows for both panel data and repeated cross-sections. In the latter case, it accommodates both random-$T_i$ sampling and fixed $T_i$-stratum sampling, as discussed in \citet{SZ20}, \citet{SX25}, and the references therein. As in \citet{SX25}, Assumption \ref{asm:sample} does not impose restrictions on compositional changes that require the distribution of $(Y_i,G_i)|T_i$ to be the same across all $T_i$. 
A natural estimator for the observed population mean outcome $\mu_t^k\equiv\mu_t^k(0)+\left(\mu_t^k(1)-\mu_t^k(0)\right)\mathds{1}\{k=1,t=T\}$ is the sample mean within the $(kt)$-cell:
\begin{align*}
    \widehat{\mu}_t^k\equiv\begin{cases}
        \frac{1}{n}\sum_{i=1}^n\frac{1}{\widehat{p}_k}\mathds{1}\{G_i=k\}Y_{it} \hspace{1.45cm}\text{ for } ~\widehat{p}_k\equiv\frac{1}{n}\sum_{i=1}^n\mathds{1}\{G_i=k\}, & \text{if panel;}\\
        \frac{1}{n}\sum_{i=1}^n\frac{1}{\widehat{\pi}_{kt}}\mathds{1}\{G_i=k,T_i=t\}Y_i ~~\text{ for } ~\widehat{\pi}_{kt}\equiv\frac{1}{n}\sum_{i=1}^n\mathds{1}\{G_i=k,T_i=t\}, & \text{if RCS}.
    \end{cases}
\end{align*}
These sample means can be plugged into $(A,b)$ in \eqref{eqn:Ab-defn}, where each observed trend $\Delta\mu_t^k\equiv\mu_t^k-\mu_{t-1}^k$ is estimated by $\widehat{\mu}_t^k-\widehat{\mu}_{t-1}^k$. Let $\big(\widehat{A}_n, \widehat{b}_n\big)$ denote the resulting estimator and
\begin{align}
    &\widehat{m}_{\tilde{\tau}}^2(\omega)\equiv m_{\tilde{\tau}}\left(\omega;\widehat{A}_n, \widehat{b}_n\right)'m_{\tilde{\tau}}\left(\omega;\widehat{A}_n, \widehat{b}_n\right),~~{m}_{\tilde{\tau}}^2(\omega)\equiv m_{\tilde{\tau}}\left(\omega;A,b\right)'m_{\tilde{\tau}}\left(\omega;A,b\right),\label{eqn:m-sq}\\
    &\hspace{2cm}\widehat{Q}_n(\tilde{\tau})~\equiv\min_{\omega\in\Delta^{K-2}}\widehat{m}_{\tilde{\tau}}^2(\omega), \quad Q(\tilde{\tau})~\equiv\min_{\omega\in\Delta^{K-2}}m_{\tilde{\tau}}^2(\omega),~~\label{eqn:Q-hat}
\end{align}
where for any $\Tilde{\tau}\in\mathbb{R}$, $\widehat{Q}_n(\tilde{\tau})$ is the sample criterion function with plugged-in $\big(\widehat{A}_n, \widehat{b}_n\big)$, whose population counterpart is $Q(\tilde{\tau})$. 

Note that $Q(\tilde{\tau})$ is a Hadamard directionally differentiable mapping of the moment function $m_{\tilde{\tau}}(\omega;A,b)$ that first squares it and then takes its minimum over $\omega\in\Delta^{K-2}$. If $$\sqrt{n}\left(\widehat{m}_{\tilde{\tau}}^2(\omega)-m_{\tilde{\tau}}^2(\omega)\right)$$ converges weakly to a Gaussian process indexed by functions of $\omega$---as it does under the assumptions in Theorem \ref{thm:asymp} below---then by the Delta method for directionally differentiable functions, the limit distribution of 
\begin{align}
    \sqrt{n}\left(\widehat{Q}_n(\tilde{\tau})-Q(\tilde{\tau})\right)\label{eqn:target-dist}
\end{align} is given by the Hadamard directional derivative of $Q(\tilde{\tau})$ applied to the limit Gaussian process. The resulting limit, denoted by $\psi(\Tilde{\tau})$, has a complex expression deferred to \eqref{eqn:psi} in Appendix \ref{appn:A}. Under the null hypothesis that $\tilde{\tau} \in \idE^+$, $Q(\tilde{\tau})=0$, and a $(1-\alpha)$-level confidence set for the elements in $\idE^+$ can then be constructed by test inversion: for $c_{\beta}(\tilde{\tau})$ denoting the $\beta$-quantile of $\psi(\Tilde{\tau})$,
\begin{align}
    \mathcal{CS}_n^{(1-\alpha)}\equiv\left\{\tilde{\tau}:\sqrt{n}\widehat{Q}_n(\tilde{\tau})\leq c_{(1-\alpha+\varsigma)}(\tilde{\tau})+\varsigma\right\},\label{eqn:CS}
\end{align}
where, following \citet{AS13}, $\varsigma>0$ is an infinitesimal uniformity factor to account for potential discontinuities in the distribution of $\psi(\tilde{\tau})$ at $c_{(1-\alpha)}(\tilde{\tau})$.
The next result establishes point-wise validity of $\mathcal{CS}_n^{(1-\alpha)}$ for all $\tilde{\tau}\in\idE^+$, after which I present a consistent estimator for the critical value via bootstrap.
\begin{theorem}
\label{thm:asymp}
    \textit{Let Assumption \ref{asm:spt} hold with non-negative weights and assume, under the sampling scheme in Assumption \ref{asm:sample}, there is a constant $0<C<\infty$ such that $0<\mathbb{E}\big[Y_{it}^2\big]<C$ in the panel data case or $0<\mathbb{E}\big[Y_i^2\big]<C$ in the RCS case. Then for any $\alpha\in(0,1)$,
    \begin{align*}
        \lim_{n\to\infty}\mathbb{P}\left(\tilde{\tau}\in\mathcal{CS}_n^{(1-\alpha)}\right)\geq1-\alpha ~\text{ for all $\tilde{\tau}\in\idE^+$.}
    \end{align*}
    }
\end{theorem}

\begin{remark}[Inequality Constraints]
    Though the inference problem of this paper only involves  equality constraints that need to be estimated, the proposed method can also be applied to linear programs with unknown \textit{inequality} constraints, in which case the moment condition in \eqref{eqn:moment-eq-compact} changes from equality to inequality:
    \begin{align*}
        m_{\tilde\tau}(\omega;A,b)\leq0.
    \end{align*}
    Instead of squaring $ m_{\tilde\tau}$, violations of inequality constraints can be measured by taking the positive part of $m_{\tilde\tau}$ using the criterion function
    \begin{align}
        \min_{\omega\in\Delta^{K-2}}[m_{\tilde\tau}(\omega;A,b)]_+,\label{eqn:new-Q}
    \end{align}
    where $[\,\cdot\,]_+\equiv\max\{\,\cdot\,, 0\}$ returns the largest positive entry of the input vector, which is again a Hadamard directionally differentiable function. Since compositions preserve directional differentiability \citep[Proposition 3.6]{S90}, the criterion function \eqref{eqn:new-Q} is also directionally differentiable and the same argument used in the proof of Theorem \ref{thm:asymp} applies.
\end{remark}

\begin{remark}[Testing Feasibility]\label{rem:test-feasible}
    Theorem \ref{thm:asymp} can be adapted to construct a test for feasibility of solutions. For example, testing whether the linear programs with nonnegativity constraints in \eqref{eqn:LP+} have a feasible solution amounts to testing
    \begin{align*}
        \exists~\omega\in\Delta^{K-2}: \Apre\omega-\bpre=0.
    \end{align*}
    A profiled criterion similar to $Q(\tilde{\tau})$ in \eqref{eqn:Q-hat} can be constructed by
    \begin{align}
        Q\equiv\min_{\omega\in\Delta^{K-2}}(\Apre\omega-\bpre)'(\Apre\omega-\bpre).\label{eqn:Q-feasibility}
    \end{align}
    One can then plug estimators for $(\Apre,\bpre)$---which are submatrix of $A$ and subvector of $b$ defined in \eqref{eqn:Ab-defn}---into $Q$ and form the sample criterion $\widehat{Q}$. The limit distribution of $\sqrt{n}(\widehat{Q}-Q)$ can be derived following a similar argument as in the proof of Theorem \ref{thm:asymp}, and feasibility of \eqref{eqn:LP+} can be rejected at a given level if the test statistic $\sqrt{n}\widehat{Q}$ exceeds the corresponding quantile of this limit distribution.

    The same idea applies to testing feasibility of the linear programs \textit{without} nonnegativity constraints in \eqref{eqn:LP} and the existence of $\varphi\in\mathbb{R}^{T_0}$ such that $\apost=[\boldsymbol{1}~~\Apre']\varphi$, which is implied by the low-rank property in Proposition \ref{prop:idM-singleton}. In these cases, however, the constraint set over which $Q$ minimizes is no longer compact and convergence as a process indexed by $\omega$ may fail. Nevertheless, under the null hypothesis that a feasible solution exists, a minimum-norm solution also exists, allowing restriction to a compact subset of feasible solutions as a proof device. Testing the feasibility of linear programs with estimated coefficients is a natural extension of the results in this paper and may be of independent interest, but a detailed treatment is left for future work.
\end{remark}

\subsection{Bootstrap Critical Values}
Due to the lack of full differentiability, standard bootstrap is inconsistent for the limiting distribution of \eqref{eqn:target-dist}; see Section 3.2 of \citet{FS19}. Procedure \ref{proc:bootstrap} below details a modified bootstrap procedure based on \citet{FS19} that uses numerical approximation to estimate the directional derivative in \eqref{eqn:phi}.

\begin{proc}[Bootstrap and Numerical Approximation of Directional Derivatives]
\label{proc:bootstrap}
\textit{
    \begin{enumerate}
        \item Draw $\{W_i\}_{i=1}^n$ iid from the exponential distribution with mean $1$ independent of the sample data $\mathcal{S}_n$ under Assumption \ref{asm:sample} and construct the bootstrap analog of $(\widehat{A}_n,\widehat{b}_n)$, denoted by $(\breve{A}_n,\breve{b}_n)$, where the bootstrap analog of the $(kt)$-cell sample mean is
{\small{ 
\begin{align*}
\breve{\mu}_t^k\equiv\begin{cases}
        \frac{1}{n}\sum_{i=1}^n\frac{1}{\breve{p}_k}\mathds{1}\{G_i=k\}W_iY_{it} \hspace{1.45cm}\text{ for } ~\breve{p}_k\equiv\frac{1}{n}\sum_{i=1}^nW_i\mathds{1}\{G_i=k\}, & \text{if panel;}\\
        \frac{1}{n}\sum_{i=1}^n\frac{1}{\breve{\pi}_{kt}}\mathds{1}\{G_i=k,T_i=t\}W_iY_i ~~\text{ for } ~\breve{\pi}_{kt}\equiv\frac{1}{n}\sum_{i=1}^nW_i\mathds{1}\{G_i=k,T_i=t\}, & \text{if RCS},
    \end{cases}
\end{align*}}}Note that for panel data, the same multiplier $W_i$ is used for all time series of the individual $i$, $(Y_{i1},...,Y_{iT})$. Let $\breve{m}^2_{\tilde{\tau}}(\omega)\equiv m_{\tilde{\tau}}\big(\omega; \breve{A}_n, \breve{b}_n)'m_{\tilde{\tau}}\big(\omega; \breve{A}_n, \breve{b}_n\big)$.
\item  Numerically approximate the directional derivative of $\min_{\omega\in\Delta}(\,\cdot\,)$, denoted by $\phi_{m^2_{\tilde\tau}}$ and defined in \eqref{eqn:phi},  in direction $g_n(\omega)=\sqrt{n}\big(\breve{m}^2_{\tilde{\tau}}(\omega)-\widehat{m}^2_{\tilde{\tau}}(\omega)\big)$ by
\begin{align}
    \widehat{\phi}_{m^2_{\tilde\tau}}\big(g_n(\omega)\big)\equiv\frac{1}{s_n}\left\{\min_{\omega\in\Delta}\left(\widehat{m}^2_{\tilde{\tau}}(\omega)+s_ng_n(\omega)\right)-\min_{\omega\in\Delta}\widehat{m}^2_{\tilde{\tau}}(\omega)\right\},\label{eqn:num-approx}
\end{align}
where the step size $s_n\to0$ and $s_n\sqrt{n}\to\infty$.
\item Obtain the bootstrap critical value
    \begin{align}  \widehat{c}_\beta(\tilde{\tau})\equiv\inf\bigg\{c:\mathbb{P}\bigg(\left.\widehat{\phi}_{m^2_{\tilde\tau}}\big(g_n(\omega)\big)\leq c \,\,\right|\, \mathcal{S}_n\bigg)\geq \beta\bigg\},\label{eqn:bootstrap-cv}
    \end{align}
    as an estimator for $c_\beta(\tilde{\tau})$, the $\beta$-quantile of the limit distribution of \eqref{eqn:target-dist}.
    \end{enumerate}
}
\end{proc}

The following result establishes consistency of the bootstrap critical value \eqref{eqn:bootstrap-cv}.

\begin{prop}\label{prop:boostrap-cv}
    \textit{Let the assumptions in Theorem \ref{thm:asymp} hold. For $\tilde{\tau}\in\idE^+$, if the limit distribution of \eqref{eqn:target-dist} is continuous and increasing at  $c_\beta(\tilde{\tau})$, then $\widehat{c}_\beta(\tilde{\tau})\xrightarrow[]{p}c_\beta(\tilde{\tau}).$}
\end{prop}

\begin{remark}[Alternative Perturbations]
    While Procedure~\ref{proc:bootstrap} is theoretically valid, its direct implementation has practical limitations. In \eqref{eqn:num-approx}, the objective function of the first minimization problem, $\widehat{m}_{\tilde{\tau}}^2(\omega) + s_n g_n(\omega)$,
    is \emph{non-convex} in $\omega$ as $g_n(\omega)$ enters with a minus sign in front of $\widehat{m}_{\tilde{\tau}}^2(\omega)$. The  simulation study in Section~\ref{sec:sim} relies on a quadratic programming solver that is computationally efficient but requires convex objectives.\footnote{See \citet{osqp} and the companion package \href{https://osqp.org/}{\texttt{OSQP}} for more information.} To ensure convexity, the code directly perturbs the estimated matrices $(\widehat{A}_n,\widehat{b}_n)$ in the direction $g_n^A=\sqrt{n}\big(\breve{A}_n-\widehat{A}_n\big)$ and $g_n^b=\sqrt{n}\big(\breve{b}_n-\widehat{b}_n\big)$,
    and replaces the non-convex objective $\widehat{m}_{\tilde{\tau}}^2(\omega) + s_n g_n(\omega)$ in \eqref{eqn:num-approx} with
    \begin{align}
        m_{\tilde{\tau}}\left(\omega; \widehat{A}_n + s_n g_n^A,\, \widehat{b}_n + s_n g_n^b\right)'
        m_{\tilde{\tau}}\left(\omega; \widehat{A}_n + s_n g_n^A,\, \widehat{b}_n + s_n g_n^b\right),
        \label{eqn:alt-pert}
    \end{align}
    which remains convex in $\omega$. This alternative perturbation is asymptotically equivalent to Procedure \ref{proc:bootstrap}: its first-order expansion at $(\widehat{A}_n,\widehat{b}_n)$ divided by $s_n$ coincides with that of $\widehat{m}_{\tilde{\tau}}^2(\omega) + s_n g_n(\omega)$  given in \eqref{eqn:rem-alt-pert}. I retain Procedure \ref{proc:bootstrap} in the text for its clarity in relation to Theorem \ref{thm:asymp}. 
\end{remark}

\section{Simulation}
\label{sec:sim}
I assess the finite-sample performance of the inference method introduced in Section \ref{sec:inf} in a Monte Carlo simulation. I also evaluate the robustness of SPT under convex weights in DGPs that satisfy different identifying assumptions, and compare it with DID, SC, and SDID. The DGPs are constructed based on the subsample of the Current Population Survey (CPS) data used in the placebo study of \citetalias{AAHIW21}, which provides repeated cross-sections of log weekly earnings for women from 1979 to 2018 ($T=40$) across $K=50$ states.\footnote{See Section~VI.1 of \citetalias{AAHIW21} (\citeyear{AAHIW21}, Online Appendix) for details on how this subsample is constructed.} In each Monte Carlo replication, a placebo treated state is randomly selected according to the treatment assignment model of \citetalias{AAHIW21}, which correlates  assignment with state-specific minimum wage laws. The remaining $49$ states form the never treated group. The last observed year 2018 is the only post-treatment period, where the placebo treated state receives no treatment so that the true effect $\tau=0$.

I consider two sets of DGPs. The first explores settings where the parallel trends assumption may or may not hold (DGP-1 and DGP-2 of Table \ref{tbl:sim}). In each replication, a repeated cross-section within each state-$k$-year-$t$ cell is iid drawn from a normal distribution:
\begin{align*}
    Y_i\,|\,G_i=k,T_i=t ~~\stackrel{d}{\sim}~~\mathcal{N}\big(\mu_t^k,(\sigma_t^k)^2\big),
\end{align*}
where $\sigma_t^k$ is equal to the sample standard deviation in each state-year cell of the CPS data. DGP-1 and DGP-2 have the same  second moments $\sigma_t^k$ and differ only in the population means $\mu_t^k$, as the two DGPs are designed to vary which identifying assumption holds true. DGP-1 is such that the parallel trends assumption holds in all time periods, where the values for $\mu_t^k$ satisfy \eqref{eqn:pt-rcs-weight} with the parallel-trend implied weight $\omega^{\texttt{PT,RCS}}$. DGP-2 is such that the parallel trends assumption holds only in the pre-period, where the values for $\mu_t^k$ and the underlying true weight $\omega\neq\omega^{\texttt{PT,RCS}}$ yield parallel pre-trends but not parallel post-trends.

\begin{table}[H]
\centering
\caption{\vspace{-.1cm} Simulation Results}
\label{tbl:sim}

\renewcommand{\arraystretch}{.8}
\setlength{\aboverulesep}{-1pt}
\setlength{\belowrulesep}{0pt}
\setlength{\abovetopsep}{-1pt}
\setlength{\belowbottomsep}{-1pt}
\setlength{\heavyrulewidth}{-1pt}

\resizebox{1\columnwidth}{!}{
\begin{tabular}{lccccccccccc}
{} &
\multicolumn{3}{c}{Bias} &
\multicolumn{4}{c}{CI length} &
\multicolumn{4}{c}{Coverage} \\
\cmidrule(lr){2-4}\cmidrule(lr){5-8}\cmidrule(lr){9-12}
\rowcolor{pink!15}
DGP & DID & SC & SDID & DID & SC & SDID & SPT & DID &  SC & SDID & SPT \\
\cline{1-12}
DGP-1 & $0.042$ & $0.038$ & $0.035$ & $0.216$ & $0.257$ &  $0.252$ & $0.824$ & $94.6\%$ & $99.2\%$ & $97.6\%$ &  $100\%$\\
DGP-2 & $0.112$ & $0.138$ & $0.116$ & $0.216$ & $1.090$ & $1.055$ & $1.614$ & $49.8\%$ & $99.8\%$ & $100\%$ & $100\%$\\
\cline{1-12}
DGP-3 &  $0.110$ & $0.100$ & $0.092$ & $0.547$ & $0.482$ & $0.446$ & $1.890$ & $95\%$ & $95\%$ & $95.4\%$ & $100\%$ \\
DGP-4 &  $4.691$ & $4.694$ &  $4.666$ & $0.547$ & $4.640$ & $4.956$ & $6.952$ & $0\%$ & $0\%$ & $0\%$ & $98.4\%$ \\
\cline{1-12}
\end{tabular}
}
\vspace{.2cm}
\begin{tablenotes}[flushleft,para]
\scriptsize
\item \begin{minipage}{\linewidth}
  \setlength{\parskip}{0pt}%
  \linespread{1}\selectfont
  Results are averaged across 500 Monte Carlo replications. The total sample size in each replication is $n = 845{,}920$ ($\raisebox{0.2ex}{\tiny$\sim$}423$ per state-year cell), where for control states, the simulation sample size of each state-year cell equals its 2018 sample share among all controls, multiplied by the total number of control observations in 2018 (based on the original CPS subsample); for the treated state, the number of observations in each year equals its 2018 value. For DID, I estimate an event-study regression of $Y_i$ on $D_i$, $T_i$, and their interaction, with 2017 as the baseline year. The coefficient on the interaction term for $D_i=1$ and $T_i=2018$ is the DID estimate for the treatment effect of the treated state in the post-period (2018). This is numerically equivalent to a double difference-in-means estimator: the difference in the treated state's 2018 and 2017 mean outcomes minus the corresponding difference for the control states' pooled mean outcomes. The $95\%$ confidence interval (CI) is obtained by adding and subtracting $1.96$ times the heteroskedastic-robust standard error, without clustering for DGP-1 and DGP-2 (RCS) and with individual-level clustering for DGP-3 and DGP-4 (panel data). SC and SDID are implemented using the \hyperlink{https://synth-inference.github.io/synthdid/index.html}{\texttt{synthdid}} package of \citetalias{AAHIW21}, with CIs constructed via the placebo variance estimator. For SPT, in each iteration $500$ candidate values are tested, each with $1000$ bootstrap replications, as detailed in Section \ref{sec:inf}. For DGP-1 and DGP-2, the average run-time per implementation in seconds (with the corresponding run-time for DGP-3 and DGP-4 in parentheses) for each method is $4.541$ ($2.473$) for DID, $28.345$ ($6.709$) for SC, $19.514$ ($11.881$) for SDID, and $106.850$ ($471.301$) for SPT.
  \end{minipage}
\end{tablenotes}
\end{table}

As expected, under DGP-1, when the parallel trends assumption holds in all periods, DID performs well, as do the other methods (first row of Table \ref{tbl:sim}). In contrast, when the post-trend is not parallel under DGP-2, DID fails to cover the null effect in over $50\%$ of all replications (second row of Table \ref{tbl:sim}). Because the parallel trends assumption still holds in the pre-periods, examining the event-study coefficients prior to treatment would not reveal this violation. In this case, both SC and SDID remain robust by widening their confidence intervals (CIs): the biases of DID, SC, and SDID are of similar magnitude, where $0.11$ is roughly the amount of violation of parallel post-trend introduced in DGP-2. Consistent with the findings in \citetalias{AAHIW21}, SDID tends to have a smaller bias than SC. CIs for SC and SDID are constructed using the placebo variance estimator defined in Algorithm 4 of \citetalias{AAHIW21} for the case of a single treated unit. This estimator tends to be large when the scale and dispersion of the population means $\mu_t^k$ are large, which is the case in DGP-2.\footnote{In DGP-1, the population means $\mu_t^k$ for the control states are set equal to the corresponding sample means in the original CPS subsample, and the treated state’s mean $\mu_t^1$ is constructed to satisfy the parallel trends assumption as in \eqref{eqn:pt-rcs-weight}. However, because the control means based on the CPS data exhibit little variation, any alternative weight $\omega$ different from the parallel-trends-implied weight $\omega^{\texttt{PT,RCS}}$ in \eqref{eqn:pt-rcs-weight} would generate only a small post-treatment violation. To create a meaningful violation of parallel post-trends in DGP-2 large enough to exceed half the length of the DID confidence interval under DGP-1 (where $0.217/2 \approx 0.11$), I increase the dispersion of the control states’ post-treatment population means in DGP-2.} The CI under SPT is dependent on the scale of the underlying population means by construction and consequently also widens under DGP-2. In contrast, the asymptotic variance of the DID estimator depends only on the idiosyncratic noise $\sigma_t^k$, which is held constant across DGP-1 and DGP-2. Thus, the DID standard errors, when averaged across many individuals, exhibit little variation across the two DGPs.

The second set of DGPs explores a setting where SC-based methods may select an incorrect weight (DGP-3 and DGP-4 of Table \ref{tbl:sim}). In each replication, a set of panel data within each state $k$ is iid drawn from a multivariate normal distribution:
\begin{align*}
    (Y_{i1},...,Y_{iT})|G_i=k~~\stackrel{d}{\sim}~~\mathcal{N}\big(\mu_k,\Sigma\big),
\end{align*}
for $\Sigma$ the rescaled $T\times T$ covariance matrix used in the placebo study of \citetalias{AAHIW21}, obtained by fitting an AR(2) model to the CPS data and assumed to be homoskedastic across states.\footnote{In \citetalias{AAHIW21}, the covariance matrix models the distribution of state-level \emph{sample means}. To generate individual-level data within each state, I rescale this covariance matrix by multiplying it by $n_k$, the sample size within each state $k$. This ensures that the variance of the individual-level outcomes is consistent with the sampling variation implied by the linear factor model in \eqref{eqn:lfm}.} Similar to the first set of DGPs, both DGP-3 and DGP-4 have the same $\Sigma$ but different mean vectors $\mu_k=\big[\mu_1^k,...,\mu_T^k\big]'\in\mathbb{R}^T$. Let $L^{\texttt{SDID}}$ be the $K\times T$ matrix of factors stacking $\lambda_t'\gamma_k$ across states and years used in the placebo study of \citetalias{AAHIW21}, obtained by fitting a rank-$4$ matrix to the matrix of sample cell means from the CPS data; see their Eq. (11). Let $L^{\texttt{SDID}}_{k\cdot}$ denote the $k$-th row of $L^{\texttt{SDID}}$. DGP-3 is such that $\mu_k=(L^{\texttt{SDID}}_{k\cdot})'$ for control states, and for the placebo treated state $k=1$,  $\mu_1=\sum_{k=2}^K\omega_k (L^{\texttt{SDID}}_{k\cdot})'$ for a randomly drawn convex weight $\omega$.\footnote{I regenerate the treated state's factors because the $K\times T$ factor matrix $L^{\texttt{SDID}}$ used in \citetalias{AAHIW21} is such that none of its rows is a convex combination of the other rows, in which case SPT with convex weight is refuted by the data, returning an empty confidence set when implementing the inference method in Section \ref{sec:inf}.} In this case, all methods perform well: as the $50\times 40$ matrix $L^{\texttt{SDID}}$ only has rank $4$, there are many weights that can reproduce the treated unit's factor across all periods, as is also reflected in the wide CI under SPT (third row of Table \ref{tbl:sim}).

DGP-4 is such that only three control states with $\mu_k=L^{\texttt{SDID}}_{k\cdot}$ are relevant for reproducing the treated state's factors, $\mu_1=\frac{1}{3}\sum_{k\in \mathcal{I}_{\texttt{rel}}} (L^{\texttt{SDID}}_{k\cdot})'$, where  $\mathcal{I}_{\texttt{rel}}$ collects the indices of the $3$ relevant controls, which are randomly selected in each replication. The remaining $46$ control states $k\in [K]\setminus\big(\mathcal{I}_{\texttt{rel}}\cup\{1\}\big)$ have the same factors as the treated state only in the pre-treatment periods, and their post-treatment factors are all set equal to the treated state's post-treatment factor \textit{minus $5$}. In this case, essentially all convex weights can achieve an equally good pre-treatment fit. DID will select the convex weight equal to the control states' shares in the never-treated sample; SC under the simplex constraint tends to select a sparse convex weight, but its chosen weight may not be the correct one that only assigns non-zero weight to the $3$ relevant controls; SDID with an $\ell^2$ penalty tends to select a dispersed convex weight, assigning approximately equal weight to all control states. However, these selected weights underweight the $3$ relevant control states and overweight the remaining $46$ control states that are systematically different from the treated state's post-treatment factor by a constant of $5$. Indeed, in the last row of Table \ref{tbl:sim}, the biases of DID, SC, and SDID are all close to $5$, and their CIs fail to cover the null effect in all of the $500$ Monte Carlo replications. In contrast, only SPT remains robust, as it accommodates the multiplicity of weighting schemes that can reproduce the treated state's pre-treatment factors, thereby capturing the true underlying weights associated with the three relevant controls.

\section{Conclusion}
\label{sec:conclusion}
This paper develops a unifying identification framework under the \textit{Synthetic Parallel Trends} (SPT) assumption formally stated in Assumption \ref{asm:spt}, which generalizes and connects the key identifying assumptions behind DID, SC, and their variants. By accounting for all weights that reproduce the treated unit's pre-treatment trends, SPT yields a partially identified set for the treatment effect, nesting existing methods as special cases and robust to violations of their respective identifying assumptions. Leveraging its linear-programming characterization, I propose an inference procedure that constructs a valid confidence set for the identified set. Simulation evidence shows that the proposed method achieves nominal coverage even when the identifying assumptions of existing methods are violated.

The analysis offers new insight into the identifying restrictions of popular empirical strategies and provides new inference methods under weak conditions. Several extensions are natural directions for future research. The identified set under SPT with either affine or convex weights can be tightened by incorporating additional credible restrictions; covariates may be particularly informative in this context, for instance by requiring the weights to match on covariates as in the SC literature. For statistical inference, it would be useful to investigate the role of efficient weighting matrices in the criterion function \eqref{eqn:Q-hat}, which currently uses the identity matrix implicitly. Finally, establishing reasonable conditions for uniform validity of the confidence set is an important question for further work.

\newpage
\bibliographystyle{ecta-fullname} 
\bibliography{synPT.bib}  

\newpage
\begin{appendix}
\numberwithin{asm}{section}
\section{Proofs of Main Results}\label{appn:A}
\begin{proof}[Proof of Proposition \ref{prop:id}]
If $\mu$ is a counterfactual trend such that Assumption \ref{asm:spt} is satisfied, i.e., there exists $\boldsymbol{1}'\omega=1$ such that in the post-period $T$,
\begin{align}
    \mu=\sum_{k=2}^K\omega_k\Delta\mu_T^k(0) \label{eq:proof-id1}
\end{align}
and in pre-periods $t\in\{2,...,T_0\}$,
\begin{align}
    \Delta\mu_t^1(0)=\sum_{k=2}^K\omega_k\Delta\mu_t^k(0). \label{eq:proof-id2}
\end{align}
Then \eqref{eq:proof-id2} implies $\omega\equiv[\omega_2, \dots,\omega_K]'$ satisfies $\Apre\omega=\bpre$, and \eqref{eq:proof-id1} implies $\mu\in\idM$. To prove sharpness, take any $\mu^*\in\idM$. Then there is an affine $\omega^*$ associated with $\mu^*$ such that $\Apre\omega^*=\bpre$ and $\mu^*=\apost'\omega^*$ by definition of $\idM$, which implies that Assumption \ref{asm:spt} is satisfied for the counterfactual trend $\mu^*$ with weight $\omega^*$.
\end{proof}

\begin{proof}[Proof of Proposition \ref{prop:idM-singleton}]
    First focus on the maximization problem in \eqref{eqn:LP} as the primal; the same argument applies to the minimization problem. Its dual is $\min_{\varphi\in\mathbb{R}^{T_0}} ([1 ~~\bpre']\varphi)$ such that $[\boldsymbol{1}~~\Apre']\varphi=\apost$. Under Assumption \ref{asm:spt}, the primal is feasible. If it attains a bounded optimal value, then the dual is also feasible and attains the same optimal value by strong duality \citep[e.g., see][]{BV04}. This implies the last row of the trend matrix in \eqref{eqn:trend-mat} is an affine transformation of the other $T_0$ rows.

    Suppose the last row of the trend matrix in \eqref{eqn:trend-mat} is an affine transformation of the other $T_0$ rows. Then the dual program is feasible and, under Assumption \ref{asm:spt}, attains a bounded optimal value (because if the dual is feasible but unbounded, then by duality the primal is infeasible, which is ruled out by Assumption \ref{asm:spt}). Strong duality again implies the primal also attains the same bounded optimal value. It then follows that $\idM$ is bounded if and only if the affine transformation condition holds.

    It remains to show when $\idM$ is bounded, the lower and upper bounds coincide so that it is in fact a singleton. Observe that both the maximization and minimization problems in \eqref{eqn:LP} have the same constraints and the same objective vector $\apost$, and so do their respective dual programs. Fix an arbitrary feasible dual solution $\varphi^*$. Then for any feasible $\omega$,
    \begin{align*}
        \apost'\omega=(\varphi^*)'\begin{bmatrix}
            \boldsymbol{1}'\\
            \Apre
        \end{bmatrix}\omega=(\varphi^*)'\begin{bmatrix}
            1\\
            \bpre
        \end{bmatrix},
    \end{align*}
    where the first equality follows from feasibility of $\varphi^*$ so that $\apost=[\boldsymbol{1}~~\Apre']\varphi^*$ and the second equality follows from feasibility of $\omega$. Since the choice of $\varphi^*$ is arbitrary, the second equality implies that $\apost'\omega$ is fixed for all primal-feasible $\omega$. Thus the maximal and minimal optimal values coincide.
\end{proof}

\begin{proof}[Proof of Proposition \ref{prop:scid}]
\label{proof:scid}
   Since $\lambda_t'\gamma_k$ is nonrandom, $\mu_t^k(0)=\mathbb{E}[\lambda_t'\gamma_k]=\lambda_t'\gamma_k$ and 
    \begin{align}
        \left|\sum_{k=2}^K\mathbb{E}\big[\widehat{\omega}_k^{\texttt{SC}}\big]\mu_t^k(0)-\mu_t^1(0)\right|=&\left|\sum_{k=2}^K\mathbb{E}\left[\widehat{\omega}^{\texttt{SC}}_k\right]\lambda_t'\gamma_k-\lambda_t'\gamma_1\right|=\left|\mathbb{E}\left[\sum_{k=2}^K\widehat{\omega}^{\texttt{SC}}_k\lambda_t'\gamma_k-\lambda_t'\gamma_1\right]\right|\notag\\
        &\hspace{3.5cm}\leq\mathbb{E}\left[\left|\lambda_t'\left(\sum_{k=2}^K\widehat{\omega}^{\texttt{SC}}_k\gamma_k-\gamma_1\right)\right|\right]\label{eqn:bound1}
    \end{align}
    for all $t\in\{1,...,T\}$.
    Adapting the notation of \citetalias{ADH10}, let $\lambda_{\texttt{P}}$ denote the $T_0\times F$ matrix where the rows stack $\lambda_t'$ for $t\leq T_0$ and $\epsilon_{k,\texttt{P}}$ denote the $T_0\times 1$ vector that, for each $k\in\{1,...,K\}$, stacks $\epsilon_{kt}$ for $t\leq T_0$. Then under the factor model \eqref{eqn:lfm} and Assumption \ref{asm:ADH}, 
    \begin{align}
        0&=\lambda_{\texttt{P}}\biggl(\sum_{k=2}^K\widehat{\omega}_k^{\texttt{SC}}\gamma_k-\gamma_1\biggr)+\sum_{k=2}^K\widehat{\omega}_k^{\texttt{SC}}\epsilon_{k,\texttt{P}}-\epsilon_{1,\texttt{P}}\notag\\
        \implies&\lambda_t'\biggl(\sum_{k=2}^K\widehat{\omega}_k^{\texttt{SC}}\gamma_k-\gamma_1\biggr)=\lambda_t'(\lambda_{\texttt{P}}'\lambda_{\texttt{P}})^{-1}\lambda_{\texttt{P}}'\biggl(\epsilon_{1,\texttt{P}}-\sum_{k=2}^K\widehat{\omega}_k^{\texttt{SC}}\epsilon_{k,\texttt{P}}\biggr), \label{eqn:noise-pre}
    \end{align}
    for any $t\in\{1,...,T\}$, implying that
    \begin{align*}
        \eqref{eqn:bound1}=&\mathbb{E}\left[ \left|\lambda_t'(\lambda_{\texttt{P}}'\lambda_{\texttt{P}})^{-1}\lambda_{\texttt{P}}'\biggl(\epsilon_{1,\texttt{P}}-\sum_{k=2}^K\widehat{\omega}_k^{\texttt{SC}}\epsilon_{k,\texttt{P}}\biggr)\right|\right]\\
        \leq&\mathbb{E}\left[ \left|\lambda_t'(\lambda_{\texttt{P}}'\lambda_{\texttt{P}})^{-1}\lambda_{\texttt{P}}'\epsilon_{1,\texttt{P}}\right|\right]+\mathbb{E}\left[ \left|\lambda_t'(\lambda_{\texttt{P}}'\lambda_{\texttt{P}})^{-1}\lambda_{\texttt{P}}'\biggl(\sum_{k=2}^K\widehat{\omega}_k^{\texttt{SC}}\epsilon_{k,\texttt{P}}\biggr)\right|\right].
    \end{align*}
Under Assumption \ref{asm:ADH}(ii), the largest element of $\lambda_t\in\mathbb{R}^F$ is upper bounded in absolute value by $\overline{\lambda}<\infty$ and the smallest eigenvalue of $\frac{1}{T_0}\lambda_\texttt{P}'\lambda_\texttt{P}$ is lower bounded by $\underline{\xi}>0$. It then follows from the argument in \citetalias{ADH10} (p. 504) that, by Cauchy-Schwarz inequality and eigenvalue decomposition of $(\lambda_\texttt{P}'\lambda_\texttt{P})^{-1}$,
\begin{align*}
    \big(\lambda_t'(\lambda_\texttt{P}'\lambda_\texttt{P})^{-1}\lambda_s\big)^2\leq \big(\lambda_t'(\lambda_\texttt{P}'\lambda_\texttt{P})^{-1}\lambda_t\big)\cdot\big(\lambda_s'(\lambda_\texttt{P}'\lambda_\texttt{P})^{-1}\lambda_s\big)\leq \frac{\lambda_t'\lambda_t}{T_0\underline{\xi}}\cdot\frac{\lambda_s'\lambda_s}{T_0\underline{\xi}}\leq\left(\frac{\overline{\lambda}^2F}{T_0\underline{\xi}}\right)^2
\end{align*}
for any $t,s\in\{1,...,T\}$, and
\begin{align*}
    &\mathbb{E}\left[ \left|\lambda_t'(\lambda_{\texttt{P}}'\lambda_{\texttt{P}})^{-1}\lambda_{\texttt{P}}'\biggl(\sum_{k=2}^K\widehat{\omega}_k^{\texttt{SC}}\epsilon_{k,\texttt{P}}\biggr)\right|\right]=\mathbb{E}\left[ \left|\sum_{k=2}^K\widehat{\omega}_k^{\texttt{SC}}\sum_{s=1}^{T_0}\lambda_t'(\lambda_{\texttt{P}}'\lambda_{\texttt{P}})^{-1}\lambda_s\epsilon_{ks}\right|\right]
    \\
    &\leq\mathbb{E}\left[ \left(\sum_{k=2}^K\widehat{\omega}_k^{\texttt{SC}}\left|\sum_{s=1}^{T_0}\lambda_t'(\lambda_{\texttt{P}}'\lambda_{\texttt{P}})^{-1}\lambda_s\epsilon_{ks}\right|^p\right)^{1/p}\right]\leq\left(\sum_{k=2}^K\mathbb{E}\left[\left|\sum_{s=1}^{T_0}\lambda_t'(\lambda_\texttt{P}'\lambda_\texttt{P})^{-1}\lambda_s\epsilon_{ks}\right|^p\right]\right)^{1/p}\\
    &\leq\frac{\overline{\lambda}F}{\underline{\xi}}\left(\sum_{k=2}^K\mathbb{E}\left[\left|\sum_{s=1}^{T_0}\frac{\epsilon_{ks}}{T_0}\right|^p\right]\right)^{1/p}\lesssim\frac{\overline{\lambda}F}{\underline{\xi}}(K-1)^{1/p}\cdot O\big(\max\big\{T_0^{1/p-1},T_0^{-1/2}\big\}\big),
\end{align*}
where the first inequality follows from Hölder's inequality, the second follows from $\widehat{\omega}^\texttt{SC}$ is convex and Jensen's inequality, the third follows from $|\lambda_t'(\lambda_\texttt{P}'\lambda_\texttt{P})^{-1}\lambda_s|\leq \frac{\overline{\lambda}^2F}{T_0\underline{\xi}}$, and the last follows from $|\epsilon_{ks}|$ having bounded $p$-th moment for the control $k\in\{2,...,K\}$ and Rosenthal's inequality \citep[see][for reference on Rosenthal's inequality]{IS02}. Hence for all $t\in\{1,...,T\}$,
    \begin{align}
        \mathbb{E}\left[ \left|\lambda_t'(\lambda_{\texttt{P}}'\lambda_{\texttt{P}})^{-1}\lambda_{\texttt{P}}'\biggl(\sum_{k=2}^K\widehat{\omega}_k^{\texttt{SC}}\epsilon_{k,\texttt{P}}\biggr)\right|\right]\to0 ~~\text{ as } T_0\to\infty, \label{eqn:pthmoment-ctl}
    \end{align}
    and under the additional assumption that $|\epsilon_{1t}|$ also has bounded $p$-th moment, the same argument from \citetalias{ADH10} applies to show that for all $t\in\{1,...,T\}$,
    \begin{align}
        \mathbb{E}\left[ \left|\lambda_t'(\lambda_{\texttt{P}}'\lambda_{\texttt{P}})^{-1}\lambda_{\texttt{P}}'\epsilon_{1,\texttt{P}}\right|\right] \to0 ~~\text{ as } T_0\to\infty \label{eqn:pthmoment-trt}.
    \end{align}
    Hence for all $t\in\{1,...,T\}$, and in particular for $t\leq T_0$, the inequality in \eqref{eqn:bound1} implies $$\left|\sum_{k=2}^K\mathbb{E}\big[\widehat{\omega}_k^{\texttt{SC}}\big]\mu_t^k(0)-\mu_t^1(0)\right|\,\xrightarrow[]{}\,0,$$
    so for all $2\leq t\leq T_0$, $\left|\sum_{k=2}^K\mathbb{E}\big[\widehat{\omega}_k^{\texttt{SC}}\big]\Delta\mu_t^k(0)-\Delta\mu_t^1(0)\right|\,\xrightarrow[]{}\,0
    $,
    i.e., for $\omega^\texttt{SC}\equiv\mathbb{E}\big[\widehat{\omega}^\texttt{SC}\big]$,
    \begin{align}
     \Apre\omega^\texttt{SC}-\bpre\,\xrightarrow[]{}\,\boldsymbol{0}.\label{eqn:pre-valid-sc}   
    \end{align} 
    In addition, since \eqref{eqn:pthmoment-ctl}-\eqref{eqn:pthmoment-trt} also hold for $t=T$, $\big|\sum_{k=2}^K\mathbb{E}\big[\widehat{\omega}_k^{\texttt{SC}}\big]\mu_T^k(0)-\mu_T^1(0)\big|\,\xrightarrow[]{}\,0$, so
    \begin{align}
        \apost' \omega^\texttt{SC}-\Delta\mu_T^1(0) \,\xrightarrow[]{}\, 0.\label{eqn:post-valid-sc}
    \end{align}
    Therefore, Assumption \ref{asm:spt} is asymptotically satisfied. Let $\mathcal{W}\equiv\{\omega\in\Delta^{K-2}:\Apre\omega=\bpre\}$ be the set of feasible solutions. Then under Assumption \ref{asm:asymp-Apre},
    \begin{align*}
         &\inf_{\delta\in\idM^+}\left|(\lambda_T-\lambda_{T_0})'\gamma_1 -\delta\right|\leq\inf_{\delta\in\idM^+}\left|(\lambda_T-\lambda_{T_0})'\gamma_1 -\apost'\omega^\texttt{SC}\right|+\left|\apost'\omega^\texttt{SC} -\delta\right|\\
         =&\left|(\lambda_T-\lambda_{T_0})'\gamma_1 -\apost'\omega^\texttt{SC}\right|+\inf_{\omega\in\mathcal{W}}\left|\apost'\omega^\texttt{SC}-\apost'\omega\right|\quad\text{by definition of $\idM^+$}\\
         \leq&\underbrace{\left|(\lambda_T-\lambda_{T_0})'\gamma_1 -\apost'\omega^\texttt{SC}\right|}_{\to ~0 \text{ by Eq. \eqref{eqn:post-valid-sc}}}+\|\apost\|\inf_{\omega\in\mathcal{W}}\left\|\omega^\texttt{SC}-\omega\right\|,\quad~~\,\text{by Cauchy-Schwarz inequality}
    \end{align*}
    where using Hoffman error bound \citep[see also Lemma 3.2.3 of \citeauthor{FP03}, \citeyear{FP03} for a textbook reference]{H52} and \eqref{eqn:pre-valid-sc},
    \begin{align*}
        \inf_{\omega\in\mathcal{W}}\left\|\omega^\texttt{SC}-\omega\right\|\lesssim\|\Apre^\dagger\|\|\Apre\omega^{\texttt{SC}}-\bpre\|\to0.
    \end{align*}
\end{proof}

\begin{proof}[Proof of Proposition \ref{prop:sdid-id}]
That the SDID oracle weights $\Tilde{\omega}$ directly balance the nonstochastic latent factors of the treated unit and of the control units during the pre-treatment period up to a constant $\Tilde{\omega}_0$ (i.e., the part of Assumption 4 of \citetalias{AAHIW21} corresponding to Assumption \ref{asm:SDID}(i)) implies $\Apre\Tilde{\omega}-\bpre\to\boldsymbol{0}$ after first-order differencing. By Assumption \ref{asm:SDID}(ii) and convex $\nu$ summing to $1$ component-wise, 
\begin{align*}
    &\left(\mu_{N_1,T}^1(0)-\sum_{k=2}^K\Tilde{\omega}_k\mu_T^k(0)\right)-\sum_{t=1}^{T_0}{\nu}_t\left(\mu_{N_1,t}^1(0)-\sum_{k=2}^K\Tilde{\omega}_k\mu_t^k(0)\right)\\
    &=\left(\mu_{N_1,T}^1(0)-\sum_{k=2}^K\Tilde{\omega}_k\mu_T^k(0)-\Tilde{\omega}_0\right)+\sum_{t=1}^{T_0}{\nu}_t\underbrace{\left(\Tilde{\omega}_0+\sum_{k=2}^K\Tilde{\omega}_k\mu_t^k(0)-\mu_{N_1,t}^1(0)\right)}_{=o(1)\text{ by Assumption \ref{asm:SDID}(i)}}
    =o(1)\\
    &\implies\mu_{N_1,T}^1(0)-\left(\Tilde{\omega}_0+\sum_{k=2}^K\Tilde{\omega}_k\mu_T^k(0)\right)=o(1),
\end{align*}
implying $\apost'\Tilde{\omega}\to\Delta\mu_T^1(0)=\lim_{N_1\to\infty}\frac{1}{N_1}\sum_{j\in \mathcal{I}_1}(\lambda_T-\lambda_{T_0})'\gamma_j$ after first-order differencing. Finally, \eqref{eqn:set-dist-sdid} follows from a similar argument in the proof of Proposition \ref{prop:scid}.\hspace{-.1cm}
\end{proof}

\begin{proof}[Proof of Theorem \ref{thm:asymp}]
    For a generic measurable function $f\in\mathcal{F}$ of a random variable $Z_i$, let $\mathbb{G}_n[f(Z_i)]\equiv \frac{1}{\sqrt{n}}\sum_{i=1}^n \bigl(f(Z_i)-\mathbb{E}[f(Z_i)]\bigr)$ denote the empirical process indexed by the function class $\mathcal{F}$. Recall the $(kt)$-cell sample mean $\widehat{\mu}_t^k$ defined immediately above  \eqref{eqn:Q-hat}. By a first-order Taylor expansion,
    \begin{align*}
        \sqrt{n}(\widehat{\mu}_t^k-\mu_t^k)=\mathbb{G}_n\big[\Tilde{Y}_{t,i}^k\big]+o_p(1),
    \end{align*}
    where
    \begin{align*}
    \Tilde{Y}_{t,i}^k=\begin{cases}
    \mathds{1}\{G_i=k\}Y_{it}/p_k-\mathbb{E}[\mathds{1}\{G_i=k\}Y_{it}]\;\mathds{1}\{G_i=k\}/p_k^2, & \text{if panel}, \\
    \mathds{1}\{G_i=k,T_i=t\}Y_i/\pi_{kt}-\mathbb{E}[\mathds{1}\{G_i=k,T_i=t\}Y_i]\;\mathds{1}\{G_i=k,T_i=t\}/\pi_{kt}^2, & \text{if RCS}.
    \end{cases}
    \end{align*}
    Let 
    \begin{align}
        \widehat{\mu}\equiv[\widehat{\mu}_1^1,...,\widehat{\mu}_T^1,\widehat{\mu}_1^2,...,\widehat{\mu}_T^2,...,\widehat{\mu}_1^K,...,\widehat{\mu}_T^K]'\in\mathbb{R}^{TK}\label{eqn:stack-mu}
    \end{align} 
    and denote its population counterpart by ${\mu}\in\mathbb{R}^{TK}$. Then
\begin{align*}
    \sqrt{n}(\widehat{\mu}-\mu)=\mathbb{G}_n[\Tilde{Y}_i]+o_p(1),
\end{align*}
for $\Tilde{Y}_i\equiv\left[\Tilde{Y}_{1,i}^1,...,\Tilde{Y}_{T,i}^1,\Tilde{Y}_{1,i}^2,...,\Tilde{Y}_{T,i}^2,...,\Tilde{Y}_{1,i}^K,...,\Tilde{Y}_{T,i}^K\right]'\in\mathbb{R}^{TK}$. Let
\begin{align*}
     L_{\texttt{ag}}\equiv\begin{bmatrix}
        -1 & 1 & ~~0 & ~~0 & ~\cdots & 0 & ~~0 \\
        0 & -1 & ~~1 & ~~0 & ~\cdots & 0 & ~~0 \\
        \vdots & \vdots & ~~\vdots & ~~\vdots & ~\ddots & \vdots & ~~\vdots\\
        0 & 0 & ~~0 & ~~0 & \cdots & -1 & ~~1 
    \end{bmatrix}\in\mathbb{R}^{T_0\times T}
\end{align*}
be the first-order difference (lag) matrix. Recall that the entries of $A$ and $b$ defined in \eqref{eqn:Ab-defn} take the form $\Delta\mu_t^k=\mu_t^k-\mu_{t-1}^k$. Then for the $KT_0\times KT$ block diagonal matrix $\mathfrak{L}\equiv\text{diag}\{L_{\texttt{ag}},...,L_{\texttt{ag}}\}$ stacking $L_{\texttt{ag}}$ diagonally,
\begin{align*}
    \sqrt{n}\begin{bmatrix}
        \widehat{b}_n-b\\
        vec(\widehat{A}_n)-vec(A)
    \end{bmatrix}=\sqrt{n}\mathfrak{L}(\widehat{\mu}-\mu)=\mathbb{G}_n[\mathfrak{L}\Tilde{Y}_i]+o_p(1),
\end{align*}
implying that, for $m_{\tilde{\tau}}(\omega;\,\cdot\, )$ and $J(\omega)$ defined in \eqref{eqn:moment-eq-compact},
\begin{align}
    &\sqrt{n}\left\{m_{\tilde{\tau}}\left(\omega;\widehat{A}_n,\widehat{b}_n\right)'m_{\tilde{\tau}}\left(\omega;\widehat{A}_n,\widehat{b}_n\right)-m_{\tilde{\tau}}\left(\omega;A,b\right)'m_{\tilde{\tau}}\left(\omega;A,b\right)\right\}\notag\\
    &=\mathbb{G}_n\left[2\left(J(\omega)\mathfrak{L}\mu+\tilde{\tau}\boldsymbol{e}_{T_0}\right)'J(\omega)\mathfrak{L}\Tilde{Y}_i\right]+o_p(1)\label{eqn:emp-process}\\
    &\xrightarrow[]{d}\mathbb{G}\left[2\left(J(\omega)\mathfrak{L}\mu+\tilde{\tau}\boldsymbol{e}_{T_0}\right)'J(\omega)\mathfrak{L}\Tilde{Y}_i\right]\hspace{1cm}\text{in $~\ell^\infty\big(\Delta^{K-2}\big)$},\notag
\end{align}
where in the last equation, the weak convergence to a Gaussian process $\mathbb{G}[\,\cdot\,]$ in $\ell^\infty(\Delta^{K-2})$, the space of bounded functions on $\Delta^{K-2}$, follows from the fact that the class of functions $$\mathcal{F}\equiv\left\{\left(J(\omega)\mathfrak{L}\mu+\tilde{\tau}\boldsymbol{e}_{T_0}\right)'J(\omega)\mathfrak{L}\Tilde{Y}_i:\omega\in\Delta^{K-2}\right\}$$ is formed by functions linear in $\omega\in\Delta^{K-2}$ compact and hence $\mathcal{F}$ is VC-subgraph \citep[see, e.g.,][]{A94}, with an envelope given by a constant multiple of $\|\tilde{Y}_i\|$ whose square-integrability under  $\mathbb{P}$ follows from $\mathbb{E}\big[Y_{it}^2\big]<C$ in the panel data case or $\mathbb{E}\big[Y_i^2\big]<C$ in the RCS case. By \citet[Theorem 2.5.2]{VW96}, $\mathcal{F}$ is $\mathbb{P}$-Donsker, implying the weak convergence.

Next, by \citet[Lemma S.4.9]{FS19}, the mapping $\min_{\omega\in\Delta^{K-2}}(\,\cdot\,)$ from the squared moment $m_{\tilde{\tau}}^2(\,\cdot\,)\equiv m_{\tilde{\tau}}\left(\,\cdot\,;A,b\right)'m_{\tilde{\tau}}\left(\,\cdot\,;A,b\right)\in\ell^\infty(\Delta^{K-2})$ to $\mathbb{R}_+$ is Hadamard directionally differentiable at $m_{\tilde{\tau}}^2$ tangentially to $\mathcal{C}(\Delta^{K-2})$, the space of continuous functions on $\Delta^{K-2}$. Its directional derivative at $m_{\tilde{\tau}}^2(\,\cdot\,)$ in the direction $g\in\mathcal{C}(\Delta^{K-2})$  is given by
\begin{align}
    \phi_{m_{\tilde{\tau}}^2}(g)\equiv\min_{\omega^*\in\arg\min_{\omega\in\Delta^{K-2}}m_{\tilde{\tau}}^2(\omega)}g(\omega^*).\label{eqn:phi}
\end{align}
By the Delta method for directionally differentiable functions \citep[Theorem 2.1]{FS19},
\begin{align}
    \sqrt{n}\left(\widehat{Q}_n(\tilde{\tau})-{Q}(\tilde{\tau})\right)\xrightarrow[]{d}\phi_{m_{\tilde{\tau}}^2}\left(\mathbb{G}\left[2\left(J(\omega)\mathfrak{L}\mu+\tilde{\tau}\boldsymbol{e}_{T_0}\right)'J(\omega)\mathfrak{L}\Tilde{Y}_i\right]\right)\equiv\psi(\Tilde{\tau}).\label{eqn:psi}
\end{align}
    For $\tilde{\tau}\in\idE^+$, ${Q}(\tilde{\tau})=0$. If $c_{(1-\alpha+\varsigma)}(\tilde{\tau})+\varsigma$ is a continuity point of the distribution of $\psi(\tilde{\tau})$,
    \begin{align*}
        \lim_{n\to\infty}\mathbb{P}\left(\tilde{\tau}\in\mathcal{CS}_n^{(1-\alpha)}\right)=\lim_{n\to\infty}\mathbb{P}\left(\sqrt{n}\widehat{Q}_n(\tilde{\tau})\leq c_{(1-\alpha+\varsigma)}(\tilde    \tau)+\varsigma\right)\geq 1-\alpha.
    \end{align*}
    If $c_{(1-\alpha+\varsigma)}(\tilde{\tau})+\varsigma$ is a discontinuity point, then for an infinitesimal $\varsigma>0$, $c_{(1-\alpha)}(\tilde{\tau})$ is a continuity point and 
    \begin{align*}
        \lim_{n\to\infty}\mathbb{P}\left(\tilde{\tau}\in\mathcal{CS}_n^{(1-\alpha)}\right)\geq\lim_{n\to\infty}\mathbb{P}\left(\sqrt{n}\widehat{Q}_n(\tilde{\tau})\leq c_{(1-\alpha)}(\tilde    \tau)\right)=1-\alpha.
    \end{align*}
\end{proof}

\begin{proof}[Proof of Proposition \ref{prop:boostrap-cv}]
    I verify the assumptions in Theorem 3.2 of \citet{FS19}, under which Proposition \ref{prop:boostrap-cv} follows. Their Assumption 1, Assumption 3(i), and Assumption 3(iii)-(iv) hold by construction; Assumption 2 holds under Theorem \ref{thm:asymp}; Assumption 4 holds by Lemma S.3.8 in their Online Appendix. Finally, to show their Assumption 3(ii) holds, recall $\breve{m}_{\tilde{\tau}}^2(\omega)$ defined in Step 1 of Procedure \ref{proc:bootstrap} and $\widehat{m}_{\tilde{\tau}}^2(\omega)$, ${m}_{\tilde{\tau}}^2(\omega)$ defined in \eqref{eqn:m-sq}.  Let $\breve{\mu}$ be the bootstrap analog of $\widehat{\mu}$ defined in \eqref{eqn:stack-mu} under Procedure \ref{proc:bootstrap}. Then the bootstrap analog of the empirical process in \eqref{eqn:emp-process} is
    \begin{align}
        &\sqrt{n}\left\{\breve{m}_{\tilde{\tau}}^2(\omega)-\widehat{m}_{\tilde{\tau}}^2(\omega)\right\}=\sqrt{n}\left\{\breve{m}_{\tilde{\tau}}^2(\omega)-{m}_{\tilde{\tau}}^2(\omega)\right\}-\sqrt{n}\left\{\widehat{m}_{\tilde{\tau}}^2(\omega)-{m}_{\tilde{\tau}}^2(\omega)\right\}\notag\\
    &=2\left(J(\omega)\mathfrak{L}\mu+\tilde{\tau}\boldsymbol{e}_{T_0}\right)'J(\omega)\mathfrak{L}\sqrt{n}(\breve{\mu}-\mu)\label{eqn:rem-alt-pert}\\
    &\hspace{5cm}-\mathbb{G}_n\left[2\left(J(\omega)\mathfrak{L}\mu+\tilde{\tau}\boldsymbol{e}_{T_0}\right)'J(\omega)\mathfrak{L}\Tilde{Y}_i\right]+o_p(1)\notag\\
    &=\mathbb{G}_n\left[2\left(J(\omega)\mathfrak{L}\mu+\tilde{\tau}\boldsymbol{e}_{T_0}\right)'J(\omega)\mathfrak{L}(W_i-1)\Tilde{Y}_i\right]+o_p(1),\notag
    \end{align}
    where the third equality follows from \eqref{eqn:emp-process} and a similar first-order expansion applied to $\sqrt{n}\left\{\breve{m}_{\tilde{\tau}}^2(\omega)-\widehat{m}_{\tilde{\tau}}^2(\omega)\right\}$, and the last equality follows from $\sqrt{n}\left(\breve{\mu}-\mu\right)=\mathbb{G}_n[W_i\tilde{Y}_i]+o_p(1)$. Then Theorem 3.6.13 of \citet{VW96} applies, implying Assumption 3(ii) of \citet{FS19}.
\end{proof}

\numberwithin{asm}{section}
\section{Auxiliary Results}\label{appn:B}
\subsection{Weighting Across Time Periods}\label{appn:B:time-weights}
Consider the following analogy of Assumption \ref{asm:spt} with time-varying weights:
\begin{asm}[Synthetic Parallel Biases]
\label{asm:spt-time}
    \textit{
    There exists a set of weights $\{\nu_t\}_{t=1}^{T_0}$ such that $\sum_{t=1}^{T_0}\nu_t=1$ and for all $k\in\{2,...,K\}$,
    \begin{align*}
        \sum_{t=1}^{T_0}\nu_t\left(\mu_t^1(0)-\mu_t^k(0)\right)=\mu_T^1(0)-\mu_T^k(0).
    \end{align*}
    }
\end{asm}
In words, Assumption \ref{asm:spt-time} requires that the selection bias between the treated unit and the control unit $k$, $\mu_t^1(0)-\mu_t^k(0)$, has a time-varying pattern such that the post-treatment selection bias $\mu_T^1(0)-\mu_T^k(0)$ can be expressed as an affine combination of pre-treatment selection biases $\{\mu_t^1(0)-\mu_t^k(0)\}_{t=1}^{T_0}$, and this weighting scheme is stable across control units. \citet[Section 6.2]{BK23} propose a similar identification strategy by connecting the post-treatment selection biases to their pre-treatment counterparts, but assume that $\mu_T^1(0)-\mu_T^k(0)$ lies in the convex hull of all $\big\{\mu_t^1(0)-\mu_t^k(0): t\in\{1,...,T_0\}, k\in\{2,...,K\}\big\}$. In contrast, I start with the weaker assumption of affine weights and restrict the post-treatment selection bias relative to unit $k$, $\mu_T^1(0)-\mu_T^k(0)$, to live in the affine hull of pre-treatment biases specific to unit $k$ only, $\{\mu_t^1(0)-\mu_t^k(0)\}_{t=1}^{T_0}$. This not only explores the intertemporal structure within the unit-$k$ specific time series of selection biases, but also connects to the double-robustness property of the two-way weighting identification framework in \citetalias{AAHIW21}, where the counterfactual $\mu_T^1(0)$ is identified by
\begin{align}
    \mu_T^1(0)=\sum_{t=1}^{T_0}\nu_t\mu_t^1(0)+\sum_{k=2}^K\omega_k\mu_T^k(0)-\sum_{t=1}^{T_0}\sum_{k=2}^K\nu_t\omega_k\mu_t^k(0)\label{eqn:double-robust}
\end{align}
if either the time weights or the unit weights are valid. To see this, suppose Assumption \ref{asm:spt} holds and $\omega$ is a valid set of unit weights. Then for all $t\in\{1,...,T_0\}$, $\big(\mu_T^1(0)-\mu_t^1(0)\big)=\sum_{k=2}^K\omega_k\big(\mu_T^k(0)-\mu_t^k(0)\big)$ and therefore for any affine $\nu$, $\sum_{t=1}^{T_0}\nu_t\big(\mu_T^1(0)-\mu_t^1(0)\big)=\sum_{t=1}^{T_0}\nu_t\sum_{k=2}^K\omega_k\big(\mu_T^k(0)-\mu_t^k(0)\big)$, which is equivalent to \eqref{eqn:double-robust}. Similarly, if Assumption \ref{asm:spt-time} holds and $\nu$ is a valid set of time weights. Then for all $k=\{2,...,K\}$, $\big(\mu_T^1(0)-\mu_T^k(0)\big)=\sum_{t=1}^{T_0}\nu_t\left(\mu_t^1(0)-\mu_t^k(0)\right)$ and therefore for any affine $\omega$,
$\sum_{k=2}^K\omega_k\big(\mu_T^1(0)-\mu_T^k(0)\big)=\sum_{k=2}^K\omega_k\sum_{t=1}^{T_0}\nu_t\left(\mu_t^1(0)-\mu_t^k(0)\right)$, which is again equivalent to \eqref{eqn:double-robust}.

\end{appendix}
\end{document}